\documentclass[conference,final]{IEEEtran}
\IEEEoverridecommandlockouts

\usepackage{cite}
\usepackage{amsmath,amssymb,amsfonts}
\usepackage{enumerate}
\usepackage{amsthm}
\usepackage[linesnumbered]{algorithm2e}
\usepackage{graphicx}
\usepackage{textcomp}
\usepackage{thmtools}
\usepackage{complexity}
\usepackage{xcolor}
\usepackage{environ}
\usepackage{mathcommand}
\usepackage{tikz}
\usepackage{hyperref}
\usepackage[capitalize]{cleveref}
\usepackage[imakeidx, quotation, notion, scope,silent]{knowledge}

\def\BibTeX{{\rm B\kern-.05em{\sc i\kern-.025em b}\kern-.08em
    T\kern-.1667em\lower.7ex\hbox{E}\kern-.125emX}}

\hypersetup{hypertexnames=false}

\newtheorem{theorem}{Theorem}[section]
\newtheorem{lemma}[theorem]{Lemma}
\newtheorem{corollary}[theorem]{Corollary}
\theoremstyle{definition}
\newtheorem{example}[theorem]{Example}
\newtheorem{definition}[theorem]{Definition}

\newcommand{\GspanFP}[3]
    {\langle #1 \rangle _{#2}(#3)}
\newcommand{\casesland}{
\tikz[overlay,baseline]{%
    \path [fill=white,draw=black] (0.59ex, 0.59ex) circle [radius=0.8ex];
    \node () at (0.59ex, 0.59ex) {\scalebox{0.6}{$\land$}};
    }%
}

\newcommand{\caseslor}{
\tikz[overlay,baseline]{%
    \path [fill=white,draw=black] (0.59ex, 0.59ex) circle [radius=0.8ex];
    \node () at (0.59ex, 0.59ex) {\scalebox{0.6}{$\lor$}};
    }%
}

\NewEnviron{landcases}{
  \rlap{$\begin{cases}
    \BODY
  \end{cases}$}
  {\hspace{0.25ex}\casesland}
  \phantom{
    \begin{cases}
      \BODY
    \end{cases}
  }
}

\NewEnviron{lorcases}{
    \rlap{$\begin{cases}
      \BODY
    \end{cases}$}
    {\hspace{0.25ex}\caseslor}
    \phantom{
      \begin{cases}
        \BODY
      \end{cases}
    }
  }

\newcommand{\Aut}{\mathrm{Aut}}
\newcommand{\range}[1]{[#1]}
\renewcommand{\Mod}{\mathrm{Mod}}
\newcommand{\Id}{\mathrm{Id}}
\newcommand{\FPC}{\mathsf{FPC}}
\newcommand{\CPT}{\mathsf{CPT}}
\newcommand{\enc}{\mathrm{enc}}
\newcommand{\type}{\mathrm{type}}
\newcommand{\CFI}{\mathsf{CFI}}
\newcommand{\rk}{\mathsf{rk}}
\newcommand{\rank}{\mathsf{rank}}
\newcommand{\ord}{\mathsf{ord}}
\newcommand{\freeVar}{\mathrm{free}}
\newcommand{\Sym}{\mathrm{Sym}}
\newcommand{\numb}{\mathrm{number}}
\newcommand{\element}{\mathrm{element}}
\newcommand{\graph}{\mathrm{graph}}
\newcommand{\perm}{\mathrm{perm}}
\newcommand{\oA}{A^\le}
\newcommand{\stroA}{A^<}
\newcommand{\tq}{\ensuremath{\ |\ }}
\newcommand{\normle}{\trianglelefteq}
\newcommand{\im}{\mathrm{im}}

\title{Group Order Logic}

\author{
    \IEEEauthorblockN{Anatole Dahan}
    \IEEEauthorblockA{University of Cambridge\\
    Université Paris-Cité\\
    Inria, ENS-Paris}
\thanks{Funded in part by UK Research and Innovation (UKRI) under the UK government’s Horizon Europe funding guarantee: grant number EP/X028259/1.}
\thanks{Funded in part by ANR - project
QUID}
\thanks{Funded in part by ANR - project $\delta$iﬀerence}
}
\begin{document}
\maketitle

\begin{abstract}
    We introduce an extension of fixed-point logic ($\FP$) with a group-order operator ($\ord$), that computes the size of a group generated by a definable set of permutations.
    This operation is a generalization of the rank operator ($\rk$).
    We show that $\FP + \ord$ constitutes a new candidate logic for the class of polynomial-time computable queries ($\P$).
    As was the case for $\FP + \rk$, the model-checking of $\FP+\ord$ formulae is polynomial-time computable. Moreover, the query separating $\FP+\rk$ from $\P$ exhibited by Lichter in his recent breakthrough is definable in $\FP+\ord$.
    Precisely, we show that $\FP + \ord$ canonizes structures with Abelian colors, a class of structures which contains Lichter's counter-example.
    This proof involves expressing a fragment of the group-theoretic approach to graph canonization in the logic $\FP+\ord$.
\end{abstract}


\section{Introduction}
The quest to identify a logic that precisely characterizes the class of problems solvable in polynomial time ($\P$) is a central challenge in descriptive complexity. 
This question can be traced back to~\cite{chandraStructureComplexityRelational1982}, and its modern formulation was stated by Gurevich~\cite{gurevichLogicChallengeComputer1988}. 
While fixed-point logic ($\FP$) captures $\P$ on ordered structures, no logic is currently known to capture $\P$ in the general case. 
$\FP$ and its natural extensions, such as fixed-point logic with counting ($\FPC$), fail to capture $\P$~\cite{caiOptimalLowerBound1992}. 
This limitation of $\FPC$ was demonstrated using the $\CFI$-construction, a class of structures encoding the satisfiability of systems of equations over the finite field $\mathbb{F}_2$\cite{caiOptimalLowerBound1992}.

To address these limitations, extensions of $\FP$ incorporating linear-algebraic operations, such as the rank operator ($\rk$), have been proposed~\cite{dawarLogicsRankOperators2009,holmDescriptive2010,pakusaFinite2010,gradelRank2019}.
However, even $\FP + \rk$ falls short of capturing $\P$, as recently shown by Lichter~\cite{lichterSeparating2023} through a generalized class of $\CFI$-structures.

On the other hand, a lot of work has been devoted to \emph{partial} capture results, showing that on restricted classes of structures, extensions of $\FP$ are able to define all $\P$ queries. 
For instance, Grohe showed that $\FPC$ captures $\P$ on any class of structures which excludes a minor\cite{groheDescriptiveComplexityCanonisation2017}.
Those results usually rely on the definition within the logic at hand of a canonization of the structures under consideration. Indeed, for any logic extending $\FP$, the Immerman-Vardi theorem implies that the definability of a canonization on a class of structures yields the capture of $\P$ on that class.
This motivates the study of canonization algorithms, and their definability in  candidate logics for $\P$.

Parallel to this investigation, significant progress has been made in the development of efficient algorithms for graph isomorphism and canonization through a group-theoretic approach. This line of research has yielded polynomial-time isomorphism and canonization algorithms for various classes of structures~\cite{babai_monte-carlo_1979,luksIsomorphismGraphsBounded1982,babaiCanonical1983}, as well as  Babai's recent breakthrough that general graph isomorphism is solvable in quasi-polynomial time~\cite{babaiGraphIsomorphismQuasipolynomial2016}.
Notably, an early result in this area demonstrates the polynomial-time canonization of $\CFI$-structures. This result generalizes seamlessly to the broader classes of $\CFI$-constructions used in \cite{gradelRank2019}, or even in\cite{lichterSeparating2023} to separate $\FP + \rk$ from $\P$.
These findings underscore the potential of integrating group-theoretic operators into $\FP$ to extend its expressive power.

The most fundamental polynomial-time permutation group algorithm is probably Schreier-Sims algorithm\cite{simsComputational1970,simsComputation1971}, which enables, given a set of permutations, to compute the order, and recognize elements of the group generated by that set.
However, this procedure relies on stabilizing one by one the elements of the domain on which the permutations act. This process thus depends on an ordering of the domain of the permutation group at hand, and cannot be defined in an isomorphism-invariant way, while its output is isomorphism-invariant.

This situation is quite similar to the one which motivated the introduction of the $\rk$ operator: 
Gaussian elimination is inherently dependent on an ordering of the rows and columns of the matrix at hand, yet the rank of the matrix is not.
Note that, the fact that $\FP + \rk$ is strictly more expressive than $\FPC$ implies that $\FPC$ indeed cannot define Gaussian elimination, nor can it define the rank of a matrix by any other means.

In this article, we introduce a novel group-theoretic operator, $\ord$, which computes the order of a group generated by a definable set of permutations. Because the Schreier-Sims algorithm enables the computation of this operation in polynomial-time, whether a structure satisfies a formula in $\FP + \ord$ can be decided in polynomial-time (in the size of the structure). Thus, like $\rk$, the $\ord$ operator defines the isomorphism-invariant result of a polynomial-time algorithm whose computation inherently depends on an ordering of the structure at hand.

The $\ord$ operator fills an interesting space within the algebraic extensions of fixed-point logics that have been studied. In~\cite{dawarDefinability2013}, the authors consider various solvability quantifiers --- expressing the satisfiability of definable systems of equations over different algebraic structures, as abelian groups, fields, or commutative rings. In this work, the authors suggest a new \emph{permutation group membership} (GM) quantifier, that subsumes all the quantifiers considered in this article.
The $\ord$ operator defines this quantifier. Actually, the $\ord$ operator can be thought of as being to the GM quantifier what the $\rk$ operator is to the field solvability quantifier. This should be contrasted with the apparent absence of such a matrix rank operator in the context of rings.

Even with this operator available, simulating group-theoretic algorithms within $\FP + \ord$ presents significant challenges due to the reliance of those algorithms on an implicit ordering of the domain.
This difficulty is particularly pronounced in graph canonization. In this context, the group-theoretic approach consists in the computation of a canonical labeling coset. A labeling coset is a set of permutations of the domain of a structure which behaves \emph{almost} like a group.
Building a canonical labeling coset rather than a mere encoding of the canonical structure enables exploiting the structure of underlying permutation groups, but depends on the existence of an ordering of the domain. Indeed, in the ordered setting, a labeling is a reordering, and thus a permutation; while in the unordered setting, it is a bijection from the domain to an initial segment of the integers, and those bijections cannot be composed.
Schweitzer et al.~\cite{schweitzerUnifying2019} provide such a definition of labeling cosets that accounts for the distinct nature of the structure's domain and its canonical numerical representation.
However, their contributions remain algorithmic and do not provide an isomorphism-invariant representation of labeling cosets.

Our main result is that $\FP + \ord$ strictly extends the expressive power of $\FP + \rk$.
Precisely, we show that the rank of a definable matrix is definable in $\FP + \ord$, and that $\FP + \ord$ captures $\P$ on the class of structures used by Lichter to separate $\FP + \rk$ from $\P$.
This result has the direct implication that $\FP + \rk$ cannot define the order of a group given by a generating set, in the same way that $\FPC <\FP + \rk$ implies that $\FPC$ cannot define the rank of a matrix. 

We obtain this canonization result by showing that, on $\CFI$-structures, $\FP + \ord$ can simulate the graph canonization algorithm defined in~\cite{babaiCanonical1983}, using an isomorphism-invariant representation of labeling cosets.
This representation of labeling cosets relies on a notion of \emph{definable group morphisms} that we will define in \cref{sec:first_properties}.

A similar approach to the canonization of $\CFI$-structures was taken in\cite{pakusaLinear2015}, where the same algorithm is simulated in the context of $\CPT$, another candidate logic for $\P$. However, the two results differ in the way labeling cosets are represented. In $\CPT$, labeling cosets are represented as systems of equations over a finite ring.
In the case of $\FP + \ord$, our representation of labeling cosets remains purely theoretic.
While this does not seem to directly allow generalization of the classes of structures canonized, this opens the door to new representation schemes for labeling cosets.

While we do not expect $\FP + \ord$ to capture $\P$, our results suggest that $\FP + \ord$ represents a meaningful advancement in the landscape of logics for polynomial-time computation. Moreover, many other operations on permutation groups are known to be polynomial-time computable, many of them playing an important role in polynomial-time Graph Isomorphism algorithms for broader classes of graphs. This first group-theoretic logic for $\P$ sets the stage to study the relationship of those different problems in an isomorphism-invariant context.

In the following section, we define the $\ord$ operator. In \cref{sec:first_properties}, we provide a set of group-theoretic operations which are definable in $\FP + \ord$, including the morphism-definability results mentioned above. \Cref{sec:rank} is devoted to the proof that $\FP + \ord$ defines the $\rk$ operator. Finally, in \cref{sec:rk_lt_ord}, we show that $\FP + \ord$ canonizes $\CFI$-structures, thus separating $\FP + \rk$ from $\FP +\ord$ (relying on Lichter's result\cite{lichterSeparating2023}).

\section{The Group Order operator}
In this section, we introduce the $\ord$ operator. We first introduce some notations and known facts concerning logic and the capture of $\P$ that will be useful throughout the article. In a second time, we introduce our representation of permutations and sets of permutations, together with the formal definition of the $\ord$ operator.
\subsection{Preliminaries}
We denote signatures by upper-case Greek letters, structures by Fraktur symbols (e.g., $\mathfrak{A}, \mathfrak{B}, \mathfrak{C}$), and their respective domains by the corresponding Roman symbols (e.g., $A, B, C$). We assume all structures to be finite, and all signatures to be relational.

Tuples of the form $(v_1, \dots, v_l)$ are denoted by $\vec{v}$. For a set $X$, we write $|X|$ to indicate its cardinality, and for a tuple $\vec{x} = (x_1, \dots, x_k)$, we write $|\vec{x}|$ to denote its length $k$. 
Given $n\in\mathbb N$, we denote $\range{n}$ the set $\{1,2,\dots,n\}$.

The \emph{graph} of a function $f : A\to B$ is the set $\{(a,b)\in A\times B\tq f(a) = b\}$, denoted $\graph(f)$. Given a function $f : Y\to Z$ and $X\subseteq Y$, we denote $f_{\restriction X}$ the restriction of $f$ to $X$.

To keep the horizontal length of formulae reasonable, we occasionally denote large disjunctions of the form $A\lor B\lor C$ as
\[\begin{lorcases} A\\ B\\ C\end{lorcases}\]
We use a similar notation for large conjunctions.

Given a logic $\mathcal{L}$ and a signature $\Sigma$, we denote by $\mathcal{L}[\Sigma]$ the set of formulae in $\mathcal{L}$ over $\Sigma$. For a formula $\varphi \in \mathcal{L}[\Sigma]$ and a $\Sigma$-structure $\mathfrak{A}$, we denote by $\varphi(\mathfrak{A})$ the set of assignments $v : \freeVar(\varphi) \to A$ such that $(\mathfrak{A}, v) \models \varphi$. By imposing an ordering on the free variables of $\varphi$, we can view $\varphi(\mathfrak{A})$ as a $|\freeVar(\varphi)|$-ary relation over $A$. 
Usually, this ordering will be explicitly specified when defining formulae. For instance, if a formula $\varphi(x, y, z)$ is defined, we order the components of $\varphi(\mathfrak{A})$, starting with the $x$-component, followed by $y$, and then $z$. The function that maps $\mathfrak A$ to $\varphi(\mathfrak A)$ is the \emph{query} defined by $\varphi$.
Given two logics $\mathcal L,\mathcal L'$, we write $\mathcal L\le \mathcal L'$ if any query definable in $\mathcal L$ is definable in $\mathcal L'$.

All logics considered in this article are extensions of $\FPC$, whose definition relies on the notion of \emph{numerical sort}. We use the definition of the numerical sort from \cite{ottoBounded2017}, that we introduce now.
Given a signature $\Sigma$ such that $(\le)\not\in\Sigma$, and a $\Sigma$-structure $\mathfrak A$, we denote $\mathfrak A^+$ the $\Sigma\sqcup \{\le\}$-structure
$(A\cup \oA, (R^\mathfrak A)_{R\in \Sigma}, \le^{\oA})$ where $\oA = \{0, 1,\dots, |A|\}$ and $\le^{\oA}$ is the natural linear order over the set of integers $\oA$. Note that $|\oA| = |A| +1$. It will be useful later to have a notation for the prefix of the natural numbers of cardinality $|A|$. We denote $\stroA$ the set $\{0,1,\dots, |A| - 1\}$.

Intuitively, $\FP$ is the extension of first-order logic with an operator enabling the computation of the inflationary fixed-point of definable second-order functions (i.e. functions mapping relations to relations). A formal definition of $\FP$ can be found in, e.g.~\cite{groheDescriptiveComplexityCanonisation2017} or \cite{ottoBounded2017}.

We can add to any $\Sigma$-structure $\mathfrak A$ its numerical domain, and consider formulae of $\FP[\Sigma\cup\{\le\}]$, evaluating them over $\mathfrak A^+$. Since any isomorphic structures $\mathfrak{A}$ and $\mathfrak{B}$ share the same numerical domain, this extension preserves the isomorphism invariance of the logic. This is the gist of $\FPC$, whose formal definition can once again be found in~\cite{groheDescriptiveComplexityCanonisation2017} or \cite{ottoBounded2017}.
Moreover, the Immerman-Vardi theorem\cite{immermanRelationalQueriesComputable1986} ensures that all $\P$ computable arithmetic functions can be defined by $\FP$ on $\oA$. 

Note that, with this definition, arithmetic functions involving integers larger than $|A|$ require the encoding of those integers as tuples of numerical values. It is easy to see that for any $k$, one can encode integers up to $|A|^k$ as $k$-tuples of numerical values. Given a tuple of numerical variables $\vec\mu$ and an integer $m \le |A|^{|\vec\mu|}$, we write $\vec\mu\gets m$ for the assignment mapping $\vec\mu$ to the unique tuple of numerical values which encodes $m$ in $\mathfrak A$.

We follow usual conventions when writing and handling formulae. In particular, we separate \emph{domain variables}, which range over $A$, composed of \emph{domain elements}, and \emph{numerical variables}, which range over $\oA$, the set of \emph{numerical elements}. This distinction constitutes the \emph{type} of a variable.

When reasonable, we keep distinct names for domain variables (for instance $x,y,z$) and numerical variables (for instance $i,j,\mu,\nu,\lambda$). 
However, for reasons to become clear in the following chapters, the strong separation of variable symbols can often bear a cost on readability. In particular, it is often convenient to consider tuples containing variables of different types.
In the same way, we usually keep distinct names for variables, and their values. When this seems to hinder readability, we may break this rule.

The \emph{type} of a tuple $\vec x$ of variables, denoted $\type(\vec x)$ is the unique word $w\in\{\element,\numb\}^*$ such that $x_i$ is a domain variable iff $w_i = \element$.
We often need to consider the set underlying all potential valuations of a tuple $\vec x$, and therefore denote $A^{\vec x}$ (or $A^{\type(\vec x)}$) the set $\prod_{i = 1}^{|\vec x|} A^{\type(\vec x)_i}$, where $A^{\element} := A$ and $A^{\numb} := \oA$.

We also allow types instead of arity in the definition of signatures. For instance, if a relation symbol $R$ has type $(\numb,\numb,\element)$, an interpretation of $R$ on $A$ is a subset of $\oA\times\oA\times A$. 
Finally, we overload this notation to relations themselves, so that if $X$ is a relation over $A$, $\type(X)$ is the unique type-word such that $X\subseteq A^{\type(X)}$.

An \emph{isomorphism} between two $\Sigma$-structures $\mathfrak A,\mathfrak B$ is a function $f : A\to B$ such that, for any relation $R\in \Sigma$ and any tuple $\vec a\in A^{\type(R)}$,
\[\vec a\in R(\mathfrak A) \iff f^*(\vec a)\in R(\mathfrak B)\]
where $f^*(\vec a)$ is the vector $\vec b$ defined by
\[ b_i := \begin{cases}
    f(a_i)&\text{if }a_i\in A\\
    a_i&\text{if }a_i\in \oA
\end{cases}\]
Given a logic $\mathcal L$ and a class $\mathcal C$ of $\Sigma$-structures, $\mathcal L$ \emph{canonizes} structures in $\mathcal C$, if there are formulae $(\varphi_R)_{R\in\Sigma}$, each with $|\type(R)|$ free numerical variables, such that, for any structure $\mathfrak A\in\mathcal C$,
\[(\stroA, (\varphi_R(\mathfrak A)_{R\in\Sigma}))\simeq \mathfrak A\]
If $\mathcal C$ consists of structures over several signatures, $\mathcal L$ canonizes $\mathcal C$ if for any signature $\Sigma$, $\mathcal L$ canonizes the class of all $\Sigma$-structures in $\mathcal C$.
This definition of canonization is quite restrictive, compared for instance to \cite[Definition 3.3.2]{groheDescriptiveComplexityCanonisation2017}. This choice is made purely to enhance clarity.

Given a logic $\mathcal L$ and a class of structures $\mathcal C$, $\mathcal L$ is said to \emph{capture} $\P$ \emph{on} $\mathcal C$ if, for any polynomial-time query $Q$ over the signature $\Sigma$, there is a formula $\varphi\in \mathcal L$ such that, for any $\Sigma$-structure $\mathfrak A\in\mathcal C$, $Q(\mathfrak A) = \varphi(\mathfrak A)$.
It is a direct consequence of the Immerman-Vardi theorem that, if $\mathcal L\ge \FP$ canonizes structures in $\mathcal C$, then $\mathcal L$ captures $\P$ on $\mathcal C$.

\subsection{Representation of sets of permutations in first-order logic}

Given a set $X$, we denote $\Sym(X)$ the group of permutations over $X$, i.e. the set of all bijections $f : X\to X$. Given $S \subseteq \Sym(X)$, we denote $\langle S\rangle$ the minimal group $G\le \Sym(X)$ which contains $S$. If $G$ is a group, we write $H\le G$ when $H$ is a subgroup of $G$ (i.e. it is a group contained in $G$). In such a case, a \emph{(left) coset} of $H$ in $G$ is a set of the form $gH := \{g\cdot h\tq h\in H\}$, for some $g\in G$. The set of cosets of $H$ in $G$ forms a partition of $G$ into $|G|/ |H|$ classes, of $|H|$ elements each.\\

As outlined in the introduction, we aim to define an operator which enables the computation of $|\langle S\rangle|$, when given a representation of $S$ as input. We first introduce such a representation of permutations and sets of permutations in first-order logic.

\begin{definition}
A formula $\varphi(\vec s,\vec t)$ \emph{defines} a permutation $\sigma\in\Sym(A^{\vec s})$ on $\mathfrak A$ if $\varphi(\mathfrak A) = \graph(\sigma)$.
\end{definition}

Conversely, given a relation $R\subseteq X\times X$ which is the graph of a permutation on $X$, we denote $\perm(R)$ this unique permutation.

Because $\FP$ does not provide a seamless way to represent sets, we represent sets of permutations as enumerations of graphs of permutations: 

\begin{definition}
A formula $\varphi(\vec p,\vec s,\vec t)$ \emph{defines $S$ binding $\vec p$ in $\mathfrak A$} if
\[ S = \{\sigma\in \Sym(A^{\vec s})\tq \exists \vec a\in A^{\vec p},  \graph(\sigma) = \varphi(\mathfrak A,\vec a)\}.\]
In such a case, we denote the group $\langle S\rangle$ as $\GspanFP{\varphi}{\vec p.\vec s.\vec t}{\mathfrak A}$.
\end{definition}
Note that we do not require that $\varphi$ defines a permutation on $\mathfrak A$ for all valuations of $\vec p$, but rather consider the set of permutations defined by $\varphi$ for \emph{some} valuation of $\vec p$. This can be seen as a way to work around the fact that, given such a formula $\varphi$, it is undecidable whether on all structures $\mathfrak A$ and for all valuations of $\vec p$, $\varphi$ defines a permutation.
Notice also that we have used dots in the definition of $\GspanFP{\varphi}{\vec p.\vec s.\vec t}{\mathfrak A}$ as to separate the three ``blocks'' of variables bound by this operator: the parameters of the enumeration, the pre-image of the permutation, and its image.
When clear from context, we may omit the two last blocks in the subscript.

There is one last obstacle to the definition of the $\ord$ operator. Recall that, given a formula $\varphi(\vec p,\vec s,\vec t)$, we expect $(\ord_{\vec p.\vec s.\vec t} \varphi)$ to represent $|\GspanFP{\varphi}{\vec p.\vec s.\vec t}{\mathfrak A}|$.
However, this value may exceed any polynomial bound on $|A|$:
\begin{example}
    Consider the formula
    \[\varphi(p_1,p_2,s,t) := \begin{lorcases}
        s = p_1  \land t = p_2\\ 
        t = p_1 \land s = p_2\\ 
        s \ne p_1\land s \ne p_2 \land s = t
    \end{lorcases}
    \]
    For any $\mathfrak A$ and $a,b\in A$, $\varphi(\mathfrak A,a,b) = \graph((a\ b))$, where $(a\ b)$ is the permutation fixing all points in $A\setminus\{a,b\}$ and mapping $a$ to $b$ and $b$ to $a$. 
    Such a permutation is called a transposition, and it is a well known fact that all finite permutations on a set can be written as a finite product of transpositions.
    
    $\varphi$ binding $p_1,p_2$ defines the set of all transposition $(a\ b)$ on the domain of the structure. Therefore, $|\GspanFP{\varphi}{p_1,p_2. s.t}{\mathfrak A}| = |\Sym(A)| = |A|!$.
\end{example}
On the other hand, this example is maximal: for any type $T$, $|\Sym(A^T)| = |A|^{|T|}!$, and as such, the binary representation of $|G|$ for any group $G\le \Sym(A^T)$ requires at most $\log(|A|^{|T|}!) \le |A|^{|T|}\log(|A|^{|T|})\le |A|^{2|T|}$. 

We are limited by the usual representation of integers in fixed-point logic as numerical values, which equates to unary encoding of integers: we can only consider integers bounded by a polynomial, while, with binary encoding, it is expected that the \emph{length} of the numbers at hand be polynomially bounded.
However, there is a straight-forward representation of the binary encoding of integers in fixed-point: as numerical relations.
That is, the number $N = \sum_{i = 1}^\ell w_i 2^{\ell - i}$ is represented as a numerical relation $R$ such that $R(x)$ holds iff $w_x = 1$.
If $\ell$ is bounded by $|A|^k$ for some constant $k$, it suffices to consider a $k$-ary relation $R$ to encode $N$.

Therefore, for any $T$, and any formula $\varphi(\vec p,\vec s,\vec t)$ with $\type(\vec s) = \type(\vec t) = T$, $N := |\GspanFP{\varphi}{\vec p. \vec s. \vec t}{\mathfrak A}|$ can be represented as a relation of type $\numb^{2|T|}$, representing the binary encoding of $N$.
This yields the following definition of the $\ord$ operator:
\begin{definition}
    \label{def:ord_operator}
Given a $\Sigma$-formula $\varphi(\vec p,\vec s,\vec t)$, with $\type(\vec s) = \type(\vec t) = T$, and a $\Sigma$-structure $\mathfrak A$,
$(\ord_{\vec p.\vec s.\vec t} \varphi)$ is a $2|T|$-ary numeric relation that encodes $|\GspanFP{\varphi}{\vec p. \vec s.\vec t}{\mathfrak A}|$, that is, for any $\vec\mu\in (\oA)^{2|T|}$,
$(\ord_{\vec p.\vec s.\vec t} \varphi)$ holds on $(\mathfrak A,\vec\mu)$ iff the $\left(\sum_{i = 1}^{2|T|} |A|^{2|T|-i}\mu_i\right)$-th bit of the binary decomposition of $|\GspanFP{\varphi}{\vec p. \vec s.\vec t}{\mathfrak A}|$ is a 1.
\end{definition}
Notice once again the use of dots in variables bound by the $\ord$ operator. We sometimes omit $\vec s$ and $\vec t$ to improve readability.

Note that, because $\FP$ captures $\P$ over ordered structures, all polynomial-time computable arithmetic properties and operations over integers in binary representation can also be defined in $\FP$ with this representation of binary integers.
\section{First properties of \texorpdfstring{$\FP + \ord$}{FP + ord}}
    \label{sec:first_properties}
We begin the study of $\FP + \ord$. This section is divided as follows: first, we present some basic facts about $\FP + \ord$. In a second subsection, we introduce the \emph{morphism formalism}, which is a key part of our labeling coset representation in \cref{sec:rk_lt_ord}. Finally, we show that $\FP + \ord$ can express any $\FP + \rk$ query.
\subsection{Model-checking, membership, union}
First, remark that the Schreier-Sims algorithm, introduced in~\cite{simsComputational1970,simsComputation1971}, precisely enables the computation in polynomial-time of the function $S\mapsto |\langle S\rangle|$. As such, it is quite easy to show that
\begin{lemma}
    For any fixed formula $\varphi\in(\FP + \ord)[\Sigma]$,
    \[\Mod(\varphi) := \{\mathfrak A\tq \mathfrak A\models \varphi\}\]
    constitutes a polynomial-time decidable class of structures
\end{lemma}
This constitutes one of the two conditions for a logic to capture $\P$ in the sense of Gurevich~\cite{gurevichLogicChallengeComputer1988}. The second one is for all $\P$-queries to be definable in $\FP + \ord$. While we do not believe this to hold, it remains unknown whether that is the case.

We now turn to the definition of group-related operations in $\FP + \ord$. First, we show that the composition and inverse of permutations are definable:
\begin{lemma}
    \label{lem:perm_equalities}
For any $\mathcal L \ge \FP$, given $\mathcal L$-formulae $\varphi_{\sigma}(\vec s,\vec t)$ and $\varphi_{\tau}(\vec s,\vec t)$, there are formulae $\varphi_{\sigma\tau}(\vec s,\vec t)$ and $\varphi_{\sigma^{-1}}(\vec s,\vec t)$, such that, for any structure $\mathfrak A$ on which $\varphi_{\sigma}$ and $\varphi_{\tau}$ define permutations,
\begin{itemize}
\item $\perm(\varphi_{\sigma\tau}(\mathfrak A)) = \perm(\varphi_\sigma(\mathfrak A))\cdot\perm(\varphi_\tau(\mathfrak A))$
\item $\perm(\varphi_{\sigma^{-1}}) = \perm(\varphi_\sigma)^{-1}$
\end{itemize}
\end{lemma}
\begin{proof}
    \begin{align*}
    \varphi_{\sigma\tau}(\vec s,\vec t)& := \exists \vec u, \begin{landcases}
    \varphi_\tau(\vec s,\vec u)\\
    \varphi_\sigma(\vec u,\vec t)
    \end{landcases}\\
    \varphi_{\sigma^{-1}}(\vec s,\vec t) &:= \varphi_\sigma(\vec t,\vec s)\qedhere
\end{align*}
\end{proof}
Besides order computation, the Schreier-Sims algorithm
enables another fundamental operation on sets of permutations: given $S\subseteq \Sym(X)$ and $\sigma\in \Sym(X)$, decide if $\sigma\in\langle S\rangle$. We now show that this operation is also definable in $\FP + \ord$:
\begin{lemma}
    \label{lem:fp_ord_membership}
Consider a pair of $(\FP + \ord)[\Sigma]$ formulae $\varphi(\vec p,\vec s,\vec t)$ and $\psi(\vec s,\vec t)$. There is a $(\FP + \ord)[\Sigma]$-sentence $(\psi\in\langle\varphi\rangle)_{\vec p. \vec s. \vec t}$ that holds on $\mathfrak A$ iff the permutation defined by $\psi$ on $\mathfrak A$ belongs to $\GspanFP{\varphi}{\vec p.\vec s.\vec t}{\mathfrak A}$.
\end{lemma}
\begin{proof}
    Fix a structure $\mathfrak A$, and let $\tau := \perm(\psi(\mathfrak A))$ and $S := \{\perm(\varphi(\mathfrak A,\vec a))\tq \vec a \in A^{\vec p}\}$.
We rely on the fact that $\tau\in\langle S\rangle$ iff $|\langle S\rangle| = |\langle S\cup\{\tau\}\rangle|$. Consider
\begin{equation}
    \label{eqn:membership}
    \chi(\vec p,b,\vec s,\vec t) := \begin{lorcases}
    b = 0 \land \varphi(\vec p,\vec s,\vec t)\\
    b = 1\land \psi(\vec s,\vec t)
    \end{lorcases}
\end{equation}
$b$ is a fresh numerical variable which enables the definition of the union of $S$ and $\{\tau\}$:
for any $\vec a\in A^{\vec p}$, $\chi(\mathfrak A,\vec a, 0) = \varphi(\mathfrak A,\vec a)$, and $\chi(\mathfrak A,\vec a, 1) = \psi(\mathfrak A)$ (for any other value of $b$, $\chi$ never holds). As such, the set of permutations defined by $\chi$ binding $\vec p,b$ is exactly
$S\cup \{\tau\}$.
Therefore, $\GspanFP{\chi}{\vec p,b.\vec s.\vec t}{\mathfrak A} = \langle S\cup\{\tau\}\rangle$. Thus, the formula
\[
    (\psi\in\langle\varphi\rangle)_{\vec p. \vec s.\vec t} := (\ord_{\vec p.\vec s.\vec t} \varphi) \hat{=} (\ord_{\vec p,b.\vec s.\vec t}\chi)
\]
fulfills the conditions of the lemma. Note that, in this definition, $\hat{=}$ denotes the equality of relations, i.e.
\[ 
    R \hat{=} R' := \forall \vec x, R(\vec x)\iff R'(\vec x)
\]
which, in this context, defines the equality of integers encoded in binary. We use such operations seamlessly, relying on the fact already mentioned below \cref{def:ord_operator}, that all arithmetic polynomial-time computable operations are $\FP$ definable.
\end{proof}
Note that we have just shown the definability of the \emph{permutation group membership} quantifier from~\cite{dawarDefinability2013}, mentioned in the introduction.

Using the same technique as in the proof of \cref{eqn:membership}, we can actually construct arbitrary unions of generating sets: given two formulae $\varphi(\vec p,\vec s, \vec t)$ and $\psi(\vec q,\vec s, \vec t)$,
\[\chi(b,\vec p,\vec q,\vec s,\vec t) := \begin{lorcases}
b = 0\land \varphi(\vec p,\vec s,\vec t)\\
b = 1\land \psi(\vec q,\vec s,\vec t)
\end{lorcases}
\]
is such that $\GspanFP{\chi}{b\vec p\vec q.\vec s. \vec t}{\mathfrak A} = \langle\GspanFP{\varphi}{\vec p.\vec s. \vec t}{\mathfrak A} \cup \GspanFP{\psi}{\vec q.\vec s. \vec t}{\mathfrak A} \rangle$. This definition can be iterated for any constant number of definable generating sets. However, to define unions of $O(|A|)$ generating sets, we rely on a different representation:
\begin{lemma}
    \label{lem:union_generating_sets}
    Given a formula $\varphi(\vec\mu,\vec p,\vec s,\vec t)$, let $G_m := \GspanFP{\varphi}{\vec p.\vec s.\vec t}{\mathfrak A, \vec \mu}$, where $\vec\mu$ is the unique tuple of numerical values encoding integer $m$.
    Then, 
    $\GspanFP{\varphi}{\vec\mu,\vec p.\vec s.\vec t}{\mathfrak A} = \left\langle \bigcup_{m<|A|^{|\vec\mu|}} G_m\right\rangle$.
\end{lemma}
With those basic tools at our disposal, we can motivate and present the morphism framework mentioned in the introduction.

\subsection{Morphism-definability}
In the previous subsection, we have seen that if a group $G$ admits a $\FP + \ord$ definable generating set, $\FP + \ord$ defines the order of $G$ (by definition of $\ord$), and the membership test on $G$ (by \cref{lem:fp_ord_membership}).

However, an important kind of group-theoretic primitives remains unexplored: operations that allow for the definition of subgroups.
The ability to construct representations of subgroups is central to many group-theoretic algorithms for graph isomorphism and canonization. Indeed, these algorithms often rely on gradually refining a large group until a generating set for the automorphism group of the graph in question has been computed.\footnote{While this process is critical to graph canonization, it also involves the problem of defining labeling cosets at each restriction step --— a topic we defer to later discussions.}

In the most general case --- given a generating set for $G$ and a membership test for $H\le G$, output a generating set for $H$ --- this operation is known to be at least as hard as Graph Isomorphism, and as such, most probably out of reach of a logic for $\P$.

On the other hand, when we additionally assume $|G|/|H|$ to be polynomially bounded, the Schreier-Sims algorithm enables a polynomial-time procedure to obtain a generating set for $H$, given a generating set for $G$ and a membership test for $H$. 
This fact is central to the polynomial-time Graph Isomorphism (and Canonization) algorithms for restricted classes of graphs introduced in~\cite{babai_monte-carlo_1979,luksIsomorphismGraphsBounded1982,babaiCanonical1983}.
Yet, this approach also appears incompatible with $\FP + \ord$ due to the need for selecting coset representatives, a process at odds with the isomorphism-invariance of $\FP + \ord$.

In this subsection, we introduce yet another restriction on $H$ which ensures that its cosets in $G$ can be represented in $\FP + \ord$: when $\FP + \ord$ can define a morphism $m : G\to K$, for some permutation group $K$, such that the \emph{kernel} of $m$ equals $H$ (all morphisms related notions are defined below).
We call such a subgroup $H$ \emph{morphism-definable} from $G$.
Such subgroups constitute a tractable scenario where cosets and intersections of subgroups can be defined within $\FP + \ord$.
However, it does not seem possible, given a morphism-definable subgroup $H$ of $G$, to define a generating set for $H$.
As such, all our results on morphism-definable subgroups rely on a shift of representation of groups: in this new context, a group is represented by a pair $(S,m)$, where $m : \langle S\rangle \to K$ for some $K$, such that $H = \{ g\in \langle S\rangle \tq m(g) = 1\}$.
For this representation of groups to be useful, we must also show that the order and membership tests of morphism-definable subgroups can be defined in $\FP + \ord$.

Recall that $H$ is a \emph{normal subgroup} of $G$, denoted $H\normle G$, if $H\le G$ and, for all $g\in G$, $gH = Hg$. A \emph{group morphism} is a function $m : G\to K$, for $G,K$ two groups, which is compatible with the operations of $G$ and $K$, i.e. $\forall g,g'\in G$,
\[ m(g\cdot g') = m(g)\cdot m(g')\]
We denote $\ker(m) := \{g\in G\tq m(g) = \Id\}$ and $\im(m) := \{m(g)\tq g\in G\}$. 
The \emph{first isomorphism theorem} states that, for any such morphism $m$,
\[|\im(m)| = \frac{|G|}{|\ker(m)|}.\]
Finally, recall that $\ker(m)\normle G$, and each normal subgroup of $G$ is realized this way: $H\normle G$ iff there is a morphism $m : G\to K$ for some group $K$ such that $\ker(m) = H$.

\begin{definition}
    \label{def:definable_morphism}
    For $T,T'$ two types, and $\varphi_m(R,\vec x,\vec y)$ a formula with $R$ a relational variable of type $T\cdot T$, and $\type(\vec x) = \type(\vec y) = T'$,
    $\varphi_m$ is said to \emph{define a morphism $m : G \to \Sym(A^{T'})$ on $\mathfrak A$}, where $G\le \Sym(A^T)$ if, for all $\sigma\in G$,
    \[\varphi_m(\mathfrak A, \graph(\sigma)) = \graph(m(\sigma))\]
    $H\normle G$ is \emph{$\mathcal L$-morphism-definable from $G$ in $\mathfrak A$} if $\mathcal L$ defines a generating set for $G$ in $\mathfrak A$, and there is a $\mathcal L$-formula $\varphi_m$ which \emph{defines a morphism} $m : G\to \Sym(A^{T'})$ on $\mathfrak A$ such that $\ker(m) = H$.
\end{definition}
In this section, we will show that, if $H\normle G$ is morphism-definable from $G$, then:
\begin{itemize}
    \item The order of $H$ is $\FP + \ord$ definable (\cref{lem:morphism_definability_repr} \emph{\ref{lem:morphism_definability_repr_ord}})
    \item Membership to $H$ is $\FP + \ord$ definable (\cref{lem:morphism_definability_repr} \emph{\ref{lem:morphism_definability_repr_memb}})
    \item Given a second subgroup $H'\normle G$ morphism-definable from $G$, their intersection $H\cap H'$ is also a morphism-definable subgroup of $G$. (\cref{corol:morph_def_grp_intersection})
\end{itemize}
The first two results ensure that morphism-definability constitutes a reasonable representation of groups, while the third is the motivation for the introduction of this new representation of groups.
That is, if we shift our representation of groups from the definability of a generating set of permutations to morphism-definability, an intersection operation on groups morphism-definable from a common group $G$ becomes definable within $\FP + \ord$.

As a first intermediate result, we show that, given a definable morphism $m : G\to K$ and a definable generating set for $G$, $\im(m)$ admits a definable generating set. 
\begin{lemma}
    \label{lem:fpc_def_image_morphism}
    Let $\mathfrak A$ be a structure. Let $\varphi(\vec p,\vec s,\vec t)$ be a formula that defines a permutation on $(\mathfrak A,\vec a)$ for all $\vec a\in A^{\vec p}$, and let $\varphi_m(R,\vec{x},\vec{y})$ be a formula defining a morphism $m : G\to \Sym(A^{\vec{x}})$, where $G := \GspanFP\varphi{\vec p.\vec s. \vec t}{\mathfrak A}$. Then, $\im(m)$ admits a definable generating set.
\end{lemma}
\begin{proof}
\[\varphi_{\im(m)}(\vec p,\vec{x},\vec{y}) :=
\varphi_m(R, \vec{x},\vec{y})[R(\vec s,\vec t) / \varphi(\vec p,\vec s,\vec t)]
\]
that is, the formula $\varphi_m$ where all occurences of $R(\vec s,\vec t)$ (for any tuples of variables $\vec s$ and $\vec t$) are substituted by $\varphi(\vec p,\vec s, \vec t)$ with variables suitably renamed to avoid capture.
\end{proof}
Moreover, given such a $v\in \im(m)$,
\[m^{-1}(v) = \{ \sigma\in G, m(\sigma) = v\}\]
is a \emph{coset} of $\ker(m)$, equal to $\sigma \ker(m)$ for any $\sigma \in m^{-1}(v)$. 
Recall that we have argued at the beginning of this section that, unlike in the computational framework, the isomorphism-invariance of $\FP + \ord$ prohibits the representation of a coset through the choice of a witness. In the case of a subgroup $H$ defined by a morphism $m$, the unique $v$ such that $m(\sigma H) = \{v\}$ constitutes a definable representative of $\sigma H$.

\begin{lemma}
    \label{lem:morphism_definability_repr}
    Let $\varphi_G(\vec p,\vec s,\vec t),\varphi_m(R,\vec x,\vec y)$ be two formulae. There are formulae $\varphi_\in(R),\varphi_{\ord}(\vec \mu)$ such that, on any structure $\mathfrak A$ on which $\varphi_m$ defines a morphism $m$ from $G := \GspanFP{\varphi_G}{\vec p.\vec s. \vec t}{\mathfrak A}$ to $\Sym(A^{\vec x})$,
    \begin{enumerate}[(i)]
        \item\label{lem:morphism_definability_repr_ord} $\varphi_{\ord}(\mathfrak A)$ is a $2|\vec{s}|$-ary numerical predicate encoding $|\ker(m)|$
        \item\label{lem:morphism_definability_repr_memb} for any $\tau\in \Sym(A^{\vec s})$, $(\mathfrak A,\graph(\tau))\models \varphi_\in$ iff $\tau\in \ker(m)$.
    \end{enumerate}
\end{lemma}
\begin{proof}
  Given $\tau\in \Sym(A^{\vec s})$, it is in $\ker(m)$ iff it is in $G$ and $m(\tau) = \Id$. Using \cref{lem:fp_ord_membership}, this is easily definable in $\FP + \ord$:
  \begin{align*}
    \varphi_\in(R) &:= (R(\vec s,\vec t)\in \langle \varphi_G\rangle)_{\vec p.\vec s.\vec t} \land \forall \vec x, \varphi_m(R,\vec x,\vec x)
    \intertext{To compute the order of $\ker(m)$, we use the fact that $|\ker(m)| = \frac{|G|}{|\im(m)|}$:}
    \varphi_{\ord}(\vec\mu)&:= \left(\frac{(\ord_{\vec p.\vec s.\vec t}\varphi_G)}{(\ord_{\vec p.\vec s.\vec t}\varphi_{\im(m)})}\right)(\vec \mu)
  \end{align*}
  where $(\frac P Q)$ denotes the result of the euclidean division of $P$ by $Q$, where $P$ and $Q$, and $(\frac P Q)$ are $2|\vec s|$-ary numerical relations encoding integers bounded by $2^{n^{2|\vec s|}}$. Note that, since euclidean division of integers (encoded in binary) is a polynomial-time computable arithmetic function, the Immerman-Vardi theorem ensures once again that the above expression is definable in $\FP + \ord$. 
\end{proof}

As mentioned above, the interest of morphism-definable subgroups stems from the ability to consider their intersection:
\begin{definition}
    \label{def:prod_morphisms}
For $m_1 : G\to X$ and $m_2 : G\to Y$ two group morphisms, let
\begin{align*}
    m_1\otimes m_2 : G&\to X\times Y\\
    g &\mapsto (m_1(g),m_2(g))
\end{align*}
\end{definition}
It is quite clear that $\ker(m_1\otimes m_2) = \ker(m_1)\cap \ker(m_2)$, and as such, $m_1\otimes m_2$ defines (from $G$) the intersection of the subgroups defined by $m_1$ and $m_2$ from $G$. It remains to show that, under a suitable permutation representation of direct products of groups, the $\otimes$ operation on morphisms is definable in $\FPC$.
Indeed, \cref{def:prod_morphisms} relies on direct product of groups, thus we must first provide a representation of direct products of groups in fixed-point logics. 

Intuitively, if $X,Y\le\Sym(A)$, we can represent $X\times Y$ as a subgroup of $\Sym(A\times\{0,1\})$ with
\[ (x,y)\cdot (a,i) := \begin{cases}
    (x(a),i)&\text{if }i = 0\\
    (y(a),i)&\text{if }i = 1.
\end{cases}
\]
However, in \cref{sec:rk_lt_ord}, we will need to construct the intersection of $O(|A|)$ different subgroups, and as such, we need a representation of polynomially large product of groups $\le \Sym(A^T)$, while keeping the type of tuples on which this representation acts constant. 
Although the operation at hand is different, this situation is similar to that of \cref{lem:union_generating_sets}.
In the following, we suppose that a \emph{family of groups} is \emph{defined by $\varphi_{\mathcal G}(\vec \mu,\vec p,\vec s,\vec t)$ over $\mathfrak A$ binding $\vec\mu$}, that is, there is a family of groups $(\mathcal G_{k})$ where, for $k\le |A|^{|\vec\mu|}$,  $\mathcal G_{k} := \GspanFP{\varphi_{\mathcal G}}{\vec p.\vec s.\vec t}{\mathfrak A,\vec \mu(k)}$, where $\vec\mu(k)$ is the unique tuple of numerical values encoding $k$ over $\mathfrak A$.
\begin{lemma}
    \label{lem:enc_prod_groups}
    Let $\varphi_{\mathcal G}(\vec \mu,\vec p,\vec s,\vec t)$ be a $\FP + \ord[\Sigma]$-formula defining a \emph{family} of groups binding $\vec \mu$, and let $\Omega := A^{\vec\mu}\times A^{\vec s}$.
    Then, $\prod_{\vec n\in (\oA)^{|\vec\mu|}} \mathcal G_{\vec n}$ is isomorphic to a subgroup of $\Sym(\Omega)$. Moreover, there is a $\FP + \ord$ formula $\varphi_{\Pi \mathcal G}$ that defines a generating set for this permutation group isomorphic to $\prod \mathcal G_{k}$.
\end{lemma}
\begin{proof}
Consider the function $\iota : \prod_{k} \mathcal G_{k}$ on $\Omega$:
\begin{align*}
    \iota : \prod_{k} \mathcal G_{k} &\to \Sym(\Omega)\\
    (\vec g)&\mapsto \left((\vec \lambda,\vec a) \mapsto (\vec\lambda, g_{\vec \lambda}(\vec a))\right)
\end{align*}
It is quite clear that $\iota$ is an injective group morphism. As such, $\iota(\prod_{k} \mathcal G_{k})$ is a permutation group representing $\prod \mathcal G_{k}$

Let us now show that it is definable in $\FPC$. As $\prod \mathcal G_{k}$ is generated by $\bigcup S_{k}$, where $S_{k}$ is a generating set for $\mathcal G_{k}$, it is sufficient to show that, for any permutation defined by $\varphi_{\mathcal G}$, we can define its action on $\Omega$ in $\FPC$:
\[ \varphi_{\Pi \mathcal G}(\vec \mu,\vec p,\vec\nu_s,\vec s,\vec \nu_t,\vec t) := \begin{landcases}
    \vec \nu_s = \vec \nu_t \\
    \begin{lorcases}
        \vec \mu = \vec \nu_s \land \varphi_{\mathcal G}(\vec\mu,\vec p, \vec s,\vec t)\\
        \vec \mu\neq \vec \nu_s \land \vec s = \vec t
    \end{lorcases}
\end{landcases}
\]
For any $k < |A|^{|\vec\mu|}$ and any $\vec c\in A^{\vec p}$, $\varphi_{\Pi \mathcal G}(\mathfrak A,\vec \mu(k),\vec c)$ defines the action of the permutation in $\mathcal G_{k}$ whose graph is $\varphi_{\mathcal G}(\mathfrak A,\vec \mu(k),\vec c)$ on $\Omega$.
As such, $\GspanFP{\varphi_{\Pi \mathcal G}}{\vec\mu\vec p}{\mathfrak A} \simeq \prod_{k} \mathcal G_{k}$. 
\end{proof}

\begin{lemma}\label{lem:fpc_defines_prod_morphisms}
    Consider two formulae $\varphi_G(\vec p,\vec s,\vec t)$ and $\varphi_m(\vec \mu,R,\vec x,\vec y)$ such that, for all structure $\mathfrak A$ and $k < |A|^{|\vec\mu|}$, $\varphi_m(\mathfrak A,\vec \mu(k))$ defines a morphism $ m_{k} : \GspanFP{\varphi_G}{\vec p}{\mathfrak A}\to \Sym(A^{\vec s})$.
    There is a formula $\varphi_{\otimes m}(R,\vec\nu_x\vec x,\vec \nu_y,\vec y)$ that defines the morphism
    \[\bigotimes_{k\le |A|^{|\vec \mu|}} m_{k} : G \mapsto \prod_{k\le |A|^{|\vec \mu|}}\im(m_{k})\]
\end{lemma}
\begin{proof}
\[
\varphi_{\otimes m}(R,\vec\nu_s,\vec s,\vec \nu_t,\vec t) :=
\begin{landcases}
    \vec \nu_s = \vec \nu_t\\
    \varphi_m(\vec\nu_s, R, \vec s, \vec t)
\end{landcases}
\qedhere\]
\end{proof}
As such, if $m_k$ morphism-defines a group $H_k\normle G$, $\bigotimes_k m_k$ morphism-defines $\bigcap H_k$. Thus,
\begin{corollary}
    \label{corol:morph_def_grp_intersection}
    If $(H_i)$ is a $\FP + \ord$ definable sequence of morphism-definable subgroups of $G$, $\bigcap H_i$ is mormhism-definable in $G$.
\end{corollary}
Note that we expect the \emph{sequence} of subgroups to be definable in the sense that the sequence of morphisms defining the groups should be given as a formula $\varphi_m$ as in \cref{lem:fpc_defines_prod_morphisms}.
\section{Defining the rank of a matrix}
    \label{sec:rank}
    The remainder of this article is devoted to the comparison of $\FP + \rk$ and $\FP + \ord$. In this section, we show that $\FP + \ord$ is at least as expressive as $\FP + \rk$, and we will show in the following section that $\FP + \ord$ expresses the query shown by Lichter to separate $\FP + \rk$ from $\P$. Together, those results imply that $\FP + \rk < \FP + \ord$.

    To show that $\FP + \rk \le \FP + \ord$, we show that any instance of the $\rk$ operator can be defined by a formula in $\FP + \ord$.
    Let us first introduce the rank operator. We denote $\mathbb F_p$ the unique finite field with $p$ elements, for $p$ a prime number.
    \begin{definition}
        \label{def:matrix_relation}
        A relation $R\subseteq X\times Y\times \mathbb N$ is said to \emph{define a matrix over $(X,Y)$} if, for all $x\in X, y\in Y$, there is a unique $m_{x,y}\in\mathbb N$ such that $(x,y,m_{x,y})\in R$.
        In such a case, $M_R := (m_{x,y})_{x\in X,y\in Y}$ is the \emph{matrix defined by $R$}.
    \end{definition}
    In such a case, and for $p$ a prime number, $M_R$ defines a linear map $f_{R,p}: \mathbb F_p^X\to \mathbb F_p^Y$, that maps $\vec u\in \mathbb F_p^{X}$ to $\vec v\in \mathbb F_p^Y$, where $v_y = \left(\sum_{x\in X} m_{x,y}u_x\right)\mod p$.
    The \emph{rank} of this linear map, denoted $\rank_p(M_R)$ is the dimension of the vector space $f_{R,p}(\mathbb F_p^X)$. Recall that, on input $M_R$, $\rank_p(M_R)$ is computable in polynomial-time through Gaussian elimination.
    \begin{definition}
        For any formula $\varphi(\vec x,\vec y,\vec\mu)$ on signature $\Sigma$ with $\vec\mu$ and $\vec\pi$ tuples of numerical variables of same length,
        \[\psi := (\rk_{\vec x.\vec y.\vec\mu}\ \varphi . \vec\pi)\]
        is a relation of type $\numb^{\min\{|\vec x|,|\vec y|\}}$ and free variables $\{\vec\pi\}\cup \freeVar(\varphi)\setminus\{\vec x,\vec y,\vec\mu\}$.

        For any $\Sigma$-structure $\mathfrak A$, any $p < |A|^{|\vec\pi|}$, and any valuation $v$ of the free variables of $\psi$: 
        \begin{itemize}
            \item if $\varphi(\mathfrak A,v)$ defines a matrix $M^{\mathfrak A,v}_\varphi$ over $(A^{\vec x}, A^{\vec y})$, $\psi(\mathfrak A,v,\vec\pi\gets p)$ consists of the unique tuple of numerical values which encodes $\rank_{p}(M^{\mathfrak A,v}_\varphi)$. 
            \item otherwise, $\psi(\mathfrak A,v) =\emptyset$
        \end{itemize}
    \end{definition}
    This definition of the $\rk$ operator is very close to~\cite[Definition 3.3.1]{pakusaFinite2010}, the main distinction being that, to avoid the handling of numerical terms, we define $(\rk \varphi)$ as a numerical relation. Note that the way its semantics are defined, there is at most one value in this relation.
    Contrary to \cite{pakusaFinite2010}, we only consider the \emph{uniform} variant of the $\rk$ operator --- when the size of the field is a variable of the operator. This is the most expressive variant of the $\rk$ operator.

    We now show that $\FP + \ord \ge \FP +\rk$:
    \begin{theorem}
        \label{thm:ord_ge_rk}
        Let $\varphi\in(\FP + \ord)[\Sigma]$, $\vec x,\vec s$ be tuples of variables and $\vec\mu,\vec\pi$ be tuples of numerical variables of the same length. There is a formula $\psi(\vec\pi)\in(\FP + \ord)[\Sigma]$ such that, for any $p < |A|^{|\vec\pi|}$, $\psi(\mathfrak A,\vec\pi\gets p) = (\rk_{\vec x.\vec s.\vec\mu} \ \varphi . \vec\pi)(\mathfrak A,\vec\pi\gets p)$.
    \end{theorem}
    \begin{proof}[Sketch of proof]
        Fix a structure $\mathfrak A$, and let $\mathcal M$ be the  matrix defined by $\varphi(\mathfrak A,v)$ over $(A^{\vec x},A^{\vec s})$.
        To ease notations, let \mbox{$I := A^{\vec x}$} and $J := A^{\vec s}$.
        Recall that $\rank_p(\mathcal M)$ is the dimension of the image of $\mathcal M$, i.e. $\im_{\mathbb F_p}(\mathcal M) := \{\mathcal M\cdot \vec X, \vec X\in \mathbb F_p^{I}\}$, and
        \[\rank_p(\mathcal M) = \log_p\left|\im_{\mathbb F_p}(\mathcal M)\right|\]

        Since the base $p$ logarithm is a $\P$-computable arithmetic function, it can be defined in $\FPC$, and it is only left to show that $|\im_{\mathbb F_p}(\mathcal M)|$ can be defined in $\FP + \ord$.

        $\im_{\mathbb F_p}(\mathcal M)$ is a subgroup of the additive group of $\mathbb F_p^{J}$, and it is the image of the group-morphism $f_{\mathcal M,p} : \mathbb F_p^{I}\to \mathbb F_p^{J}$ as defined below \cref{def:matrix_relation}.
        Thus, the theorem reduces to the definability of a generating set $\mathcal E$ for $\mathbb F^I_p$ and of the morphism $f_{\mathcal M,p}$, using \cref{lem:fpc_def_image_morphism}.
        However, because the general definition within $\FP +\ord$ of the morphism $f_{\mathcal M,p}$ is quite convoluted, we will instead provide directly a formula enumerating $f_{\mathcal M,p}(\mathcal E)$.
        The formulae defining $f_{\mathcal M,p}(\mathcal E)$ are provided in appendix, \cref{app:rank}.

    \end{proof}

        Let us mention that this proof actually generalises to broader algebraic structures than fields: for instance, for any ring\footnote{Let us recall that a ring has the same properties as a field, except for the existence of multiplicative inverses. That is, we only consider commutative rings with a multiplicative identity.} $\mathcal R$, a set of linear equations over $\mathcal R$ can be represented as a matrix $\mathcal M\subseteq \mathcal R^{I\times J}$ for some index-sets $I,J$. As long as $|\mathcal R| \le |A^k|$ for some $k$, and the addition and multiplication in $\mathcal R$ are provided, we can construct a generating set for the additive group $\im_{\mathcal R}(\mathcal M) := \{\mathcal M\cdot X, X\in \mathcal R^I\}$.
    
        This bears some importance as further candidate logics for $\P$ have been considered\cite{pakusaFinite2010,dawarDefinability2013}, introducing operators allowing to check, given such a matrix $\mathcal M$ and a tuple $Y\in \mathcal R^J$, whether there is a tuple $X\in \mathcal R^I$ such that $\mathcal M\cdot X = Y$. This is the Ring Equation Satisfiability operator (RES).
    
        Since a generating set for $\im_{\mathcal R}(\mathcal M)$ can, under the assumptions given above, be defined in $\FP + \ord$, we can check whether $Y\in \im_{\mathcal R}(\mathcal M)$ using the membership operator defined in \cref{lem:fp_ord_membership}, which shows that $\FP + \ord$ is also at least as expressive as $\FP + \mathrm{RES}$.

\section{Separation of rank and group order logics}
\label{sec:rk_lt_ord}

\knowledge{notion}
    |relative encoding
\knowledge{notion}
    |local
    |local scale

\knowledge{notion}
    |semi-local
    |semi-local scale

\knowledge{notion}
    |global
    |global scale
 
\knowledge{notion}
    |partial ordering
\knowledgenewrobustcmd
    \mapFP
    {\cmdkl {\mathsf{map}}}

\knowledgenewcommandPIE
    \mapMa
    {\cmdkl {\mathsf{map}#1#2#3}}

\knowledgenewrobustcmd
    \strA
    {{\cmdkl {\mathfrak A}}}


\knowledgenewrobustcmd
    \domA
    {\cmdkl{A}}

\newcommand
    \domAO
    {\stroA}

\knowledgenewmathcommandPIE
    \ldomA
    {\cmdkl{A#1#2#3}}

\knowledgenewmathcommandPIE
    \ldomAO
    {\cmdkl {A#2^<}}

\knowledgenewrobustcmd
    \PhiAB
    {\cmdkl {\Phi}}

\knowledgenewrobustcmd
    \precAB
    {\cmdkl {\mathrel{\preceq}}}

\knowledgenewrobustcmd
    \mAB
    {\cmdkl {\ensuremath{m}}}


\disablecommand\gamma
\knowledgerenewcommandPIE
    \gamma
    {\cmdkl{\LaTeXgamma#1#2#3}}

\disablecommand\vartheta
\knowledgerenewcommandPIE
    \vartheta
    {\cmdkl{\LaTeXvartheta#1#2#3}}

\knowledgenewcommandPIE
    \gslv
    {\cmdkl{v#1#2#3}}

\knowledgenewrobustcmd
    \encMa
    {\cmdkl{\enc}}
\knowledgenewcommandPIE\lGa{\cmdkl{\Gamma#2}#1#3}
\knowledgenewcommandPIE\olGa{\cmdkl{\Gamma^<#2}}

\knowledgenewrobustcmd\gGa{\cmdkl{\Gamma}}
\knowledgenewrobustcmd\ogGa{\cmdkl{\Gamma^<}}

\knowledge{notion}
    |labelling coset
    |Labelling coset
    |labelling cosets
    |Labelling cosets


\knowledgenewrobustcmd
    \lLaCos
    {\cmdkl{\mathcal O}}


\knowledgenewcommandPIE
    \slE
    {\cmdkl{E#1#2#3}}

\knowledgenewcommandPIE
    \slStrA
    {\cmdkl{\mathfrak A#2}}

\knowledgenewcommandPIE
    \slm
    {\cmdkl{m#1#2#3}}

\knowledge{notion}
    |\slm

\knowledgenewcommandPIE
    \slmStar
    {\cmdkl{m^*#2}}
\knowledgenewcommandPIE
    \slOm
    {\cmdkl{\Omega#2}}
\knowledgenewcommandPIE
    \kDel
    {\cmdkl{\Delta#2}}

\knowledgenewcommandPIE
    \koDel
    {\cmdkl{\Delta^<#2}}

\knowledgenewrobustcmd
    \aut
    {\cmdkl{\mathsf{aut}}}

\knowledgenewrobustcmd
    \coset
    {\cmdkl{\mathsf{coset}}}
\knowledgenewrobustcmd
    \slMorph
    {\cmdkl {\mathsf{slMorph}}}

\knowledgenewcommandPIE
    \phiRepr
    {\cmdkl{\varphi#1#2#3}}
\knowledgenewcommandPIE
    \psiRepr
    {\cmdkl{\psi#1#2#3}}

\knowledgenewrobustcmd
    \gOm
    {\cmdkl{\ensuremath{\Omega}}}
\knowledgenewrobustcmd
    \mathcalG
    {\cmdkl{\ensuremath{\mathcal G}}}

\knowledgenewrobustcmd
    \mathcalH
    {\cmdkl{\ensuremath{\mathcal H}}}
\knowledgenewrobustcmd
    \genG
    {\cmdkl{\ensuremath{\mathsf{gen}\mathcal G}}}

\knowledgenewcommandPIE
    \minit
    {\ensuremath{\cmdkl{m_{init}#1#3}}}

\knowledgenewrobustcmd
    \vinit
    {\cmdkl {\ensuremath{v_{init}}}}

\knowledgenewrobustcmd
    \initMorph
    {\cmdkl {\mathsf{initMorph}}}

\knowledgenewrobustcmd
    \initValue
    {\cmdkl {\mathsf{initValue}}}

\knowledgenewrobustcmd
    \vFP
    {\cmdkl{\mathsf{v}}}

\knowledgenewrobustcmd
    \varthetaFP
    {\cmdkl{\LaTeXvartheta}}

\knowledgenewrobustcmd
    \encFP[4]
    {\cmdkl {\enc_{#1 ; #2}#3;#4}}

\knowledgenewrobustcmd
    \compound
    {\cmdkl{\mathsf{compound}\LaTeXvartheta}}

\knowledgenewcommandPIE
    \uvarthetaMa
    {\cmdkl{\underline\LaTeXvartheta#1#2#3}}

\knowledgenewrobustcmd
    \uvarthetaFP
    {\cmdkl{\underline\LaTeXvartheta}}

In this section, we show that the $\ord$ operator is strictly more expressive than the $\rk$ operator (in the context of fixed-point logics).
We have just shown that $\FP + \rk \le \FP + \ord$, so that it is only left to exhibit a property inexpressible in $\FP + \rk$ that is definable in $\FP + \ord$.

Recently, Lichter~\cite{lichterSeparating2023} provided such a property separating $\FP + \rk$ from $\CPT$ (although whether $\FP + \rk \le \CPT$ remains unknown).
More precisely, Lichter exhibits a class of structures $\mathcal K$ and a property $\mathcal P\subseteq \mathcal K$ such that no $\FP + \rk$ formula expresses $\mathcal P$, while $\mathcal P$ is $\CPT$-definable.

The $\CPT$-definability of $\mathcal P$ stems from the fact that $\mathcal P$ is $\P$-computable, and that $\CPT$ canonizes structures in $\mathcal K$.
The ability of $\CPT$ to canonize structures in $\mathcal K$ is a direct consequence of the fact that those structures have \emph{definable abelian colors}, a notion that we will make precise in the second section of this chapter. The fact that structures with abelian colors can be canonized in $\CPT$ was proved in~\cite{pakusaLinear2015}.

We will show that $\FP + \ord$ also canonizes structures with abelian colors. First, we review the notion of abelian colors and the method used in~\cite{pakusaLinear2015} to canonize those structures in $\CPT$.
Subsequently, we adapt this method in the context of $\FP + \ord$.

\subsection{Canonizing structures with abelian colors}
A \emph{coloring} of a structure $\mathfrak A$ is a function $c : A\to \range{m}$ for some $m$. For $i\le m$, the $i$-th \emph{color-class} of $\mathfrak A$ is the set $c^{-1}(i)\subseteq A$, denoted $A_i$.
A coloring is usually represented within a structure as a total pre-order $\preceq$ (i.e. a total, transitive, reflexive binary relation), such that $c^{-1}(i)$ is the set of elements which admit a maximal $\prec$-increasing sequence of length $i$ (where $x\prec y$ iff $x\preceq y\land\lnot y\preceq x$).

An \emph{abelian group} is a commutative group. A group $G$ is said to \emph{act} on a set $X$ if we are given a morphism $m : G\to \Sym(X)$. A group action is \emph{faithful} if $m$ is injective ; it is \emph{transitive} if $\{m(g)(x), g\in G\} = X$ for some (and thus all) $x\in X$. When $G\le \Sym(X)$ we say that $G$ is transitive if its action through the identity morphism is transitive.

\begin{definition}
    \label{def:abelian_colors}
    \AP A $\Sigma$-structure with \emph{Abelian colors} is a $\Sigma\cup\{\intro *\precAB,\intro *\PhiAB\}$-structure $\mathfrak A$, where $\type(\precAB) =\element^2$ and $\type(\PhiAB) = \numb^2\element^2$, such that:
    \begin{itemize}
    \item $\reintro *\precAB$ is a total pre-order on $A$. From now on, let $\intro *\mAB$ be the number of equivalence classes of $\precAB$ ; and $A_i$ be the $i$-th equivalence class.
    \itemAP for any $i < \mAB$, $j < |A_i|$, $\reintro *\PhiAB(\mathfrak A,i,j)$ is the graph of a permutation $\intro*\gamma^i_j\in \Sym(A_i)$ (recall that $\PhiAB(\mathfrak A,i,j) = \{(s,t)\tq (i,j,s,t)\in\PhiAB^{\mathfrak A}\}$);
    \itemAP for any $i < m$, ${\intro *\lGa_i} := \{\gamma^i_j \tq j < |A_i|\}$ is an abelian, transitive permutation group over $A_i$.\footnote{Note that as a subgroup of $\Sym(\ldomA_i)$, $\lGa_i$ acts faithfully on $\ldomA_i$.}
    \end{itemize}
\end{definition}

That is, $\mathfrak A$ has abelian colors if it is equipped with a total pre-order $\preceq$ \emph{and} a relation $\Phi$ that explicitly enumerates a family of abelian transitive groups acting on each of the color-classes defined by $\preceq$. Even more so, the type of $\Phi$ implies that each of those groups is linearly ordered. Given a $\Sigma$-structure $\mathfrak A$, we call such an interpretation of $\preceq$ and $\Phi$ on $\mathfrak A$ an \emph{abelian coloring} of $\mathfrak A$. 
Note that we require that $|\Gamma_i| \le |A_i|$. This always holds because a transitive abelian group $G\le \Sym(X)$ has exactly $|X|$ elements:
\begin{lemma}\label{lem:regular_groups}
    Let $G$ be an abelian group acting faithfully and transitively on a set $X$. Then, the action of $G$ on $X$ is regular, i.e. for any $x\in X$ and $g\in G$, $g\cdot x = x \iff g = 1$. As such, for any fixed $y\in X$, the map $g\mapsto g\cdot y$ is a bijection between $G$ and $X$.
\end{lemma}
\begin{proof}
    Consider $g\in G$ such that $g\cdot x = x$. We will show that $g = 1_G$. Because the action of $G$ on $X$ is faithful, this amounts to show that, for any $y$, $g\cdot y = y$. Consider such a $y\in X$, and by transitivity, let $h\in G$ such that $h\cdot x = y$. Then,
    \begin{align*}
    g\cdot y &= (hh^{-1}g)\cdot  y
    = (hgh^{-1})\cdot y &\emph{since $G$ is abelian} \\
    &= (hg)\cdot x
    = h\cdot x = y&\qedhere
    \end{align*}
\end{proof}

The structures defined by Lichter to separate $\FP + \rk$ from $\P$ have abelian colors~\cite{lichterSeparating2023,pakusaLinear2015}.
Moreover, a direct adaptation of \cite[Theorem 5.5.1]{hodgesModel1993} implies that the canonization of structures with abelian colors reduces to the canonization of \emph{graphs} with abelian colors. Together, those results reduce the separation between $\FP + \rk$ and $\FP + \ord$ to the canonization of graphs with abelian colors in $\FP + \ord$.

\begin{theorem}
    \label{thm:ord_abelian_colors}
    $\FP + \ord$ canonizes graphs with abelian colors.
\end{theorem}
\begin{corollary}
    \label{corol:rk_lt_ord}
    $\FP + \rk < \FP + \ord$.
\end{corollary}
The remainder of this section is devoted to the proof of \cref{thm:ord_abelian_colors}.
\AP In~\cite{pakusaLinear2015}, the author introduces an algorithm to canonize structures with abelian colors. We will follow the same procedure. Before we present this algorithm, a few introductory definitions are in order.
Fix a colored graph $\mathfrak A$ and let $c$ be its coloring. We denote $\intro* \ldomA_i$ the $i$-th color-class of $\mathfrak A$, and $\intro* \slE_{i,j}$ the edges between $A_i$ and $A_j$, i.e. $\reintro* \slE_{i,j} := E \cap (A_i\times A_j\cup A_j\times A_i)$.
An \emph{ordering of $A$ consistent with $c$}  is a bijection $\sigma : A\to \stroA$ such that for any $a,b\in A$, $c(a) < c(b)\implies \sigma(a) < \sigma(b)$.
Let $\intro*\ldomAO_i := \sigma(A_i)$ for some ordering $\sigma$ consistent with $c$. Note that this definition of $\ldomAO_i$ does not depend on the choice of $\sigma$.
An \emph{ordering of $A_i$ consistent with $c$} is a bijection $\sigma : A_i\to\ldomAO_i$. It is clear that any ordering of $A$ consistent with $c$ can be decomposed in a product $\prod_{i \le m} \sigma_i$, where $\sigma_i$ is an ordering of $A_i$ consistent with $c$. 
Given a relation $R$ over $A$ and $\sigma$ an ordering of $A$, we can define the \emph{encoding of $R$ relative to $\sigma$}, denoted $R^\sigma$: $R^\sigma$ has type $\numb^{|\type(R)|}$, and $R^\sigma := \{ (\sigma^*(v_1),\sigma^*(v_2), \dots, \sigma^*(v_l))\tq (v_1,\dots,v_l)\in R\}$, where
\[\sigma^*(v) := \begin{cases}
\sigma(v)&\text{if }v\in A\\
v&\text{if }v\in\domAO
\end{cases}\]
A set of orderings $\mathcal C$ is said to \emph{canonize} $R$ if, for any $\sigma,\tau\in \mathcal C$, $R^\sigma = R^\tau$.
From now on, we only consider orderings consistent with the coloring of the structure at hand. 
Moreover, an ordering $\sigma$ is \emph{definable in $\mathcal L$} if there is a $\mathcal L$-formula such that $\varphi(\mathfrak A) = \graph(\sigma)$.
It is easy to show that, if $R$ and $\sigma$ are definable in $\mathcal L$, $R^\sigma$ is definable as well (for any logic $\mathcal L$ extending $\FPC$).

\Cref{algo:can_procedure} presents the canonization procedure used in the context of $\CPT$ in\cite{pakusaLinear2015}.
\begin{algorithm}
    \SetKwInOut{Input}{Input}
    \SetKwInOut{Output}{Output}
    \Input{$\mathfrak A = (A,E,\precAB,\PhiAB)$ a structure with Abelian colors}
    \Output{A numerical relation $E^<$ isomorphic to $E$}
    \BlankLine
    Find, for each $i\le m$, a canonical set $\mathcal O(A_i)$ of orderings of $A_i$\;\label{algo:can_procedure_l1}
    $\mathcal C := \prod_{i = 1}^m \mathcal O(A_i)$\;\label{algo:can_procedure_l2}
    \For{$(i,j)\in [m]^2$}{
        $E_{i,j}^<$, the smallest lexicographical encoding of $E_{i,j}$ which is compatible with $\mathcal C$, i.e. $\exists \sigma\in \mathcal C, E_{i,j}^\sigma = E_{i,j}^<$ \;\label{algo:can_procedure_l4}
        $\mathcal C\gets \{ \sigma\in \mathcal C, E_{i,j}^\sigma = E_{i,j}^<\}$\;\label{algo:can_procedure_l5}
    }
    \Return{$E^< := \bigcup_{i,j} E_{i,j}^<$}
    \vspace{0.5em}
    \caption{canonisation procedure}
    \label{algo:can_procedure}
\end{algorithm}
As a first remark, note that for this canonization procedure to be complete, we should also provide relations $(\precAB)^<$ and $\Phi^<$. We will actually see in \cref{lem:OAi_can_Phi} that the initial value of $\mathcal C$ already canonizes $\PhiAB$, and because by construction, all orderings in $\mathcal C$ are consistent with $c$, $\mathcal C$ canonizes $\precAB$ as well.

The structure of this algorithm in and of itself is easily definable in $\FPC$: the only control-flow mechanism is a for-loop over an ordered domain, which can obviously be implemented in $\FPC$. On the other hand, it is not obvious how to represent sets of orderings, and it is precisely in how this is achieved that we depart from~\cite{pakusaLinear2015}.
Before delving into this question, we show that the construction of sets $\mathcal O(A_i)$ on line 1 of \cref{algo:can_procedure} is $\FPC$-definable.
To ease reading, throughout this section, we fix a graph $\mathfrak A$ with abelian colors. 
\begin{lemma}
    \label{lem:local_labellings}\AP
    There is an $\FPC$-formula $\intro *\mapFP(\lambda,x,y,\mu)$ such that, for $i\le m$ and $a\in \ldomA_i$, $\mapFP(\mathfrak A,i,a)$ defines an ordering of $\ldomA_i$ (that is, a bijection between $\ldomA_i$ and $\ldomAO_i$)
\end{lemma}
\begin{proof}
    First, \cref{lem:regular_groups} ensures that the action of $\lGa_i$ on $\ldomA_i$ is regular.
    Thus, for any fixed $a\in \ldomA_i$,
    \begin{equation}
        \label{eqn:local_ordering_def}
        \gamma^i_1(a) < \gamma^i_2(a) < \dots < \gamma^i_{|A_i|}(a)
    \end{equation}
    defines a linear ordering on $\ldomA_i$.
    This ordering corresponds to a bijection between $\ldomA_i$ and $\ldomAO_i$ whose graph is easily definable in $\FPC$, using $\PhiAB$ and some basic arithmetic. The formal definition of the formula $\mapFP$ is provided in \cref{app:local_labellings}.
\end{proof}

\AP We denote $\intro *\mapMa^i_a$ the ordering defined by $\mapFP(\mathfrak A, i, a)$.
By definition of $\mapMa$, with the notations introduced above, we have, for any $a,b\in \ldomA_i$ and $\mu\in\ldomAO_i$,
\begin{equation}
    \label{eqn:mapia_defining_prop}
    \mapMa^i_a(b) = \mu \iff \gamma^i_\mu(a) = b.
\end{equation}

We follow Pakusa's notation and denote $\intro* \lLaCos(\ldomA_i)$ the set of all $\mapMa^i_a$, for $a\in A_i$.
$\mathcal O(A_i)$ is not an arbitrary set of orderings:
\begin{lemma}\label{lem:OAi_coset}
    For any $i\le m$ and $a\in \ldomA_i$, $\lLaCos(\ldomA_i) = \mapMa^i_a \lGa_i$.
\end{lemma}
\begin{proof}
    Let $a\in \ldomA_i$, and $\gamma^i_j\in \lGa_i$. Then, for any $b\in \ldomA_i$,
    \begin{align*}
        \mapMa^i_a\gamma^i_j\cdot b = \mu &\iff \gamma^i_\mu\cdot a = \gamma^i_j\cdot b &\text{by \cref{eqn:mapia_defining_prop}}\\
        &\iff ((\gamma^i_j)^{-1}\gamma^i_\mu)\cdot a = b&\\
        &\iff \gamma^i_\mu\cdot ((\gamma^i_j)^{-1}\cdot a) = b&\text{($\lGa_i$ is abelian)}\\
        &\iff \mapMa^i_{(\gamma^i_j)^{-1}\cdot a}\cdot b = \mu&
    \end{align*}
Therefore, $\lLaCos(\ldomA_i)$ is closed by multiplication on the right (i.e. precomposition) by elements of $\lGa_i$, or, said differently, $\mapMa^i_a \lGa_i \subseteq \lLaCos(\ldomA_i)$.
The transitivity of $\lGa_i$ yields the other inclusion: for $a,b\in \ldomA_i$, let $\gamma^i_j$ be the element\footnote{Recall that $\lGa_i$ is regular, and thus this element is unique} such that $\gamma^i_j\cdot b = a$.
Then, $\mapMa^i_a\gamma^i_j = \mapMa^i_{(\gamma^i_j)^{-1}\cdot a} = \mapMa^i_{b}$.
\end{proof}
We now show that $\mathcal O(A_i)$ canonizes $\PhiAB(\mathfrak A,i)$:
\begin{lemma}
    \label{lem:OAi_can_Phi}
    For any $a,b\in \ldomA_i$, the encodings of $\PhiAB(\mathfrak A,i)$ relative to $\mapMa^i_a$ and $\mapMa^i_b$ are equal.
\end{lemma}
\begin{proof}
Fix $i\le m,j\in\range{\ldomA_i}, a\in \ldomA_i$. Then, the following holds:
\begin{align*}
    \PhiAB^{\mapMa^i_a}(i,j) &= \{ (\mapMa^i_a\cdot b,\mapMa^i_a\gamma^i_j\cdot b), b\in \ldomA_i\}\\
    &= \{(\alpha, \mapMa^i_a \gamma^i_j (\mapMa^i_a)^{-1}\cdot\alpha), \alpha\in \ldomAO_i\}
\end{align*}
In words, encoding $\lGa_i$ relative to $\mapMa^i_a$ entails to enumerate the elements of $\lGa_i$ conjugated by $\mapMa^i_a$.
It happens that the conjugation actions of $\mapMa^i_a$ and $\mapMa^i_b$ on $\lGa_i$ coincide:
\begin{align*}
    (\gamma^i_j)^{\mapMa^i_a} &= \mapMa^i_a \gamma^i_j  (\mapMa^i_a)^{-1}\\
    &= \mapMa^i_a \gamma^i_j(\mapMa^i_a)^{-1}\mapMa^i_b(\mapMa^i_b)^{-1}\\
    \text{Since $\lGa_i$ is abelian,}\quad&=\mapMa^i_a(\mapMa^i_a)^{-1}\mapMa^i_b\gamma^i_j(\mapMa^i_b)^{-1}\\
&= \mapMa^i_b \gamma^i_j (\mapMa^i_b)^{-1}\\
&= (\gamma^i_j)^{\mapMa^i_b}
\end{align*}
Note that we have used the fact that $(\mapMa^i_a)^{-1}\mapMa^i_b\in \lGa_i$, which is a direct consequence of \cref{lem:OAi_coset}.
\end{proof}

It is now time to discuss the representation of sets of orderings. Indeed, \cref{algo:can_procedure} shows that the proof of \cref{thm:ord_abelian_colors} reduces to the existence of a $\FP + \ord$-definable representation of sets of orderings which enables the four following operations:
\begin{itemize}
\item The definition of $\mathcal C_0 := \prod_{i = 1}^m \lLaCos(\ldomA_i)$, as on line 2 of \cref{algo:can_procedure}.
\item Given $i,a$ and $j,b$, the definition of the set $\mathcal C^{i,j}_{a,b}$ of orderings in $\mathcal C_0$ which yield the same encoding of $\slE_{i,j}$ as $\mapMa^i_a\oplus\mapMa^j_b$.\footnote{Where, $f\oplus g$ is the minimal common extension of $f$ and $g$ (if such an extension exists).}
\item Given $\mathcal C,\mathcal C'$, the definition of $\mathcal C\cap \mathcal C'$.
\item Given $\mathcal C$, checking if $\mathcal C = \emptyset$.
\end{itemize}
Let us show how these operations enable the definition of line 4 of \cref{algo:can_procedure}.
Given $\mathcal C$, define a binary relation $\mathrm{Comp}_{\mathcal C}$ which holds on $(a,b)$ iff $\mapMa^i_a\oplus\mapMa^j_b = \sigma_{\restriction A_i\cup A_j}$ for some $\sigma\in \mathcal C$:
\[\mathrm{Comp}_{\mathcal C}(x,y) := (\mathcal C^{i,j}_{x,y}\cap \mathcal C)\ne \emptyset\]
Then, find an pair $(a,b)\in \mathrm{Comp}_{\mathcal C}$ which yields the minimal encoding of $E_{i,j}$:
\[\mathrm{min}_{\mathcal C}(x,y)\! :=\!\begin{landcases}\mathrm{Comp}_{\mathcal C}(x,y)\\
 \forall x',y' \begin{lorcases} \lnot\mathrm{Comp}_{\mathcal C}(x',y')\\
    (E_{i,j})^{\mapMa^i_x\oplus\mapMa^j_y}\!\le\!(E_{i,j})^{\mapMa^i_{x'}\oplus\mapMa^j_{y'}}
 \end{lorcases}
\end{landcases}
 \]
 (where $\le$ denotes the lexicographical comparison of numerical relations, which is obviously $\FPC$-definable).
 Because all pairs $(a,b)\in\mathrm{min}_{\mathcal C}(\mathfrak A)$ yield the same encoding of $\slE_{i,j}$, we can define
 \[E^<_{i,j}(\mu,\nu) := \exists x,y,s,t, \begin{landcases}
    \mathrm{min}_{\mathcal C}(x,y)\\
    \mapFP(i,x,s,\mu)\land
    \mapFP(j,y,t,\nu)\\
    E(s,t)
 \end{landcases}
 \]
The new set of orderings $\mathcal C'$ defined on line 5 is then $\mathcal C\cap \mathcal C^{i,j}_{a,b}$ for any $(a,b)\in \mathrm{min}_{\mathcal C}(\mathfrak A)$. 
The definition of line 2 is a direct application of the definability of $\mathcal C_0$.

In order to define such a representation of sets of orderings, let us first remark that the sets under consideration along the run of \cref{algo:can_procedure} have a strong structural property: they are of the form $\sigma\Lambda$, for some $\sigma\in \mathcal C_0$, and $\Lambda$ a subgroup of $\intro *\gGa := \prod_{i = 1}^m\lGa_i$. We call such sets \emph{labeling cosets}.
The fact that those sets are all labeling cosets is a direct consequence of the three following lemmas:
\begin{lemma}\label{corol:pi_gamma}
    For any $\pi\in \mathcal C_0$, we have $\mathcal C_0 = \pi\gGa$.
\end{lemma}
\begin{proof} 
    This is a direct consequence of \cref{lem:OAi_coset}
\end{proof}
\begin{lemma}
    For any $i,j$, there is a group $\overline\Delta_{i,j}\le\gGa$ such that $\mathcal C_{a,b}^{i,j} = \sigma\overline\Delta_{i,j}$ for any $\sigma\in\mathcal C_0$ s.t. $\sigma_{\restriction A_i\cup A_j} = \mapMa^i_a\oplus\mapMa^j_b$.
\end{lemma}
\begin{proof}
    Let $\overline\Delta_{i,j} = (\lGa_i\lGa_j\cap \Aut(E_{i,j}))\cdot\prod_{\lambda \in\range{m}\setminus\{i,j\}} \lGa_\lambda$, where
    $\Aut(E_{i,j})$ is the group of permutations of $\ldomA_i\cup\ldomA_j$ which stabilize $E_{i,j}$. The proof is straight-forward.
\end{proof}
\begin{lemma}
    Given two labeling cosets $ \sigma\Lambda, \tau\Lambda'$, $ \sigma\Lambda\cap \tau\Lambda'$ is either empty, or a labeling coset of $\Lambda\cap \Lambda'$.
\end{lemma}
\begin{proof}
    Suppose $\mathcal C := \sigma\Lambda\cap \tau\Lambda'\ne\emptyset$ and let $\rho \in\mathcal C$. Then, $\sigma\Lambda = \rho\Lambda$, $\tau\Lambda' = \rho\Lambda'$ and $\mathcal C = \rho\Lambda \cap \rho\Lambda' = \rho(\Lambda\cap\Lambda')$.
\end{proof}
At this point, let us point out that the algorithm we are aiming to define is a special case of the canonical placement-coset algorithm defined by Babai and Luks in \cite[Section 3.2]{babaiCanonical1983}.
However, in the algorithmic context, a labeling coset $\sigma\Lambda$ can be represented by an arbitrary witness $\tau\in\sigma\Lambda$ and a generating set for $\Lambda$. Such an arbitrary choice is not isomorphism-invariant.
Here, the fact that $\gGa$ is abelian comes at play: in this context, any labeling coset $\sigma\Lambda$ is such that $\Lambda \normle \gGa$, and thus there is a morphism $m_\Lambda : \gGa\to X_\Lambda$ for some group $X_\Lambda$, such that $\ker(m_\Lambda) = \Lambda$.
In such a case, $m_\Lambda$ defines a bijection between cosets of $\Lambda$ and $\im(m_\Lambda)$.

This leads to a second issue with the unordered domain: while in the algorithmic context, a labeling coset $\sigma\Lambda$ is a coset of $\Lambda$ in a group $\mathcal G$ that contains $\Lambda$ as a subgroup, this is not the case here, as orderings (i.e. bijections from $A$ to $\domAO$), unlike \emph{re}orderings (i.e. bijections from $\domAO$ to $\domAO$) cannot be composed with one another.

To overcome this, we now show that we can define in $\FPC$ a representation $\varphi$ of orderings as permutations over a fixed domain ($A^T$ for some fixed type $T$). Then, setting $\mathcal G := \langle \varphi(\pi\gGa)\rangle$, we will show that for any labeling coset $\sigma\Lambda$ considered during the run of \cref{algo:can_procedure}, $\varphi(\sigma\Lambda)$ is a coset of a morphism-definable subgroup of $\mathcal G$. With a few encoding details, this will conclude our proof of \cref{thm:ord_abelian_colors}, as \cref{corol:morph_def_grp_intersection,lem:fpc_def_image_morphism} ensure that $\FP + \ord$ defines the intersection operation and checks the emptiness of morphism-definable labeling cosets, respectively.

\subsection{Orderings as permutations}
Fix some $\pi\in \prod_{i}\lLaCos(\ldomA_i)$, and let $\pi_i := \pi_{\restriction \ldomA_i}$.
Recall that, by \cref{corol:pi_gamma}, $\pi\gGa = \prod_i\lLaCos(\ldomA_i)$. 

Let us first give an intuition of our construction: suppose we were given a ``base'' ordering $f : A\to \stroA$. Then, there is a bijection mapping any $g : A \to \stroA$ to the permutation of $A$ that maps $f$ to $g$ through composition, that is, $f^{-1} g$.
Here, while we obviously do not have access to such a fixed \emph{single} ordering, for each color class $\ldomA_i$, we have a canonical \emph{family} $\pi_i\lGa_i$ of $|\ldomA_i|$  ``base'' orderings\footnote{Recall that~\cref{lem:OAi_coset} implies that $\pi_i\lGa_i = \lLaCos(\ldomA_i)$}, and therefore any $\sigma\in\pi_i\lGa_i$ can be mapped injectively to $\lGa_i^{\ldomA_i}$:
\begin{align*}
    \intro* \phiRepr_i : \pi_i\lGa_i&\to \gGa^{A}\\
    \sigma &\mapsto \left(b\mapsto \begin{cases}
        (\mapMa^i_b)^{-1}\sigma &\text{if }b\in \ldomA_i\\
        \Id&\text{otherwise}\end{cases}\right)
\intertext{This encoding is compatible with the morphism}
    \intro *\psiRepr_i: \lGa_i &\to \gGa^{A}\\
    \LaTeXgamma&\mapsto \left(b\mapsto\begin{cases}\LaTeXgamma&\text{if }b\in \ldomA_i\\
        \Id&\text{otherwise}
    \end{cases}\right)
\end{align*}
in the sense that $\forall \sigma\in\pi_i\lGa_i,\LaTeXgamma\in\lGa_i, \phiRepr_i(\sigma\LaTeXgamma)=\phiRepr_i(\sigma)\psiRepr_i(\LaTeXgamma)$. Note that this implies that, for any $\sigma,\tau\in\pi_i\lGa_i$,
\begin{equation}
    \label{eqn:local_repr_coset}
    \phiRepr_i(\sigma)^{-1}\phiRepr_i(\tau) = \psiRepr_i(\sigma^{-1}\tau)
\end{equation}
For any $\sigma\in\pi_i\lGa_i$ and $\LaTeXgamma\in\lGa_i$, $\phiRepr_i(\sigma)$ and $\psiRepr_i(\LaTeXgamma)$ are families of elements of $\gGa$ indexed by $A$, and given $a\in A$, we denote the $a$-component of $\phiRepr_i(\sigma)$ (resp. $\psiRepr_i(\LaTeXgamma)$) by $\phiRepr_i(\sigma)_a$ (resp. $\psiRepr_i(\LaTeXgamma)_a$).
Note that, while we have defined $\phiRepr_i$ and $\psiRepr_i$ to range over $\gGa^A$, their image is actually quite restricted: first, for any $a\in A$, $\phiRepr_i(\sigma)_a$ and $\psiRepr_i(\LaTeXgamma)_a$ are in $\lGa_i$. Moreover, all the non-trivial values of $\phiRepr_i(\sigma)_a$ and $\psiRepr_i(\LaTeXgamma)_a$ are reached for $a\in \ldomA_i$. That is, morally, $\phiRepr_i$ and $\psiRepr_i$ take values in $\lGa_i^{\ldomA_i}$.
However, providing a uniform codomain to all those functions is convenient, as it allows us to combine them easily into a ``global'' representation of $\pi\gGa$ within $\gGa^A$:
\begin{align*}
    \reintro*\phiRepr : \pi\gGa &\to \gGa^A\\
    \sigma&\mapsto \prod_{i =1}^m\phiRepr_i(\sigma_{\restriction A_i})\\
    \reintro*\psiRepr: \gGa &\to \gGa^A\\
    \LaTeXgamma &\mapsto \prod_{i = 1}^m \psiRepr_i(\LaTeXgamma_{\restriction A_i})
\end{align*}
$\phiRepr$ is compatible with $\psiRepr$, hence $\phiRepr(\pi\gGa)=\phiRepr(\pi)\psiRepr(\gGa)$ is a coset of $\psiRepr(\gGa)$. Moreover, as was the case for $\phiRepr_i$ and $\psiRepr_i$:
\begin{equation}
    \label{eqn:global_repr_coset}
    \phiRepr(\sigma)^{-1}\phiRepr(\tau) = \psiRepr(\sigma^{-1}\tau)
\end{equation}
One can also easily check that $\phiRepr_i,\psiRepr_i$ are injective for all $i\le m$, and thus so are $\phiRepr$ and $\psiRepr$.
As a side note, $\gGa^A$ is not \emph{exactly} a permutation group, but a product of permutation groups. However, we can use \cref{lem:enc_prod_groups} to represent $\gGa^A$ as a subgroup of $\Sym(A\times A)$. We will mostly keep this encoding nuance implicit.
$\phiRepr(\pi\gGa)$ is a coset of $\psiRepr(\gGa)$, and we will show that $\psiRepr(\gGa)$ is morphism-definable (as defined in \cref{def:definable_morphism}), which will enable the representation of its coset $\phiRepr(\pi\gGa)$ by defining, in $\FPC$:
\begin{itemize}
    \item A generating set for a group $\mathcalG$ that contains both $\phiRepr(\pi\gGa)$ and $\psiRepr(\gGa)$ (as subsets)
    \item A morphism $\minit : \mathcalG\to \Sym(\gOm)$ such that $\ker(\minit) = \psiRepr(\gGa)$
    \item A value $\vinit\in\Sym(\gOm)$ such that $\minit^{-1}(\vinit) = \phiRepr(\pi\gGa)$.
\end{itemize}
Because $\phiRepr(\sigma) = \prod_{i = 1}^m \phiRepr_i(\sigma_{\restriction \ldomA_i})$, and for each $\sigma\in\pi\gGa$, $\sigma_{\restriction \ldomA_i} \in \lGa_i$, we have 
\[\phiRepr(\pi\gGa) \subseteq \prod_{i = 1}^m \phiRepr_i(\pi_i\lGa_i)\le \langle \bigcup_{i = 1}^m \phiRepr_i(\pi_i\Gamma_i)\rangle\]
\AP Therefore, we set $\intro*\mathcalG := \langle\bigcup_{i = 1}^m\phiRepr_i(\pi_i\lGa_i)\rangle$.
\begin{lemma}
    \label{lem:fpc_def_genG}
    There is a $\FPC$ formula $\genG$ which defines a generating set for $\iota(\mathcalG)$.
\end{lemma}
\begin{proof}
    We remind the reader that we actually define a generating set for the group $\iota(\mathcalG)$. By definition of $\mathcalG$, it is enough to build a formula $\genG$ such that, for any $i,a$,
    \[\genG(\mathfrak A,i,a) = \graph(\iota(\phiRepr_i(\mapMa^i_a))).\] 
    Consider the following formula:
\[
    \intro* \genG(p_1p_2,b_sx_s,b_tx_t) := \begin{landcases}(b_s = b_t) \\
        \begin{lorcases}
        \begin{landcases}
            i(x_s) = p_1 \\ x_t = (\mapMa^{p_1}_{b_s})^{-1}\mapMa^{p_1}_{p_2}(x_s)
        \end{landcases}\\
        i(x_s)\ne {p_1} \land x_s = x_t
    \end{lorcases}
\end{landcases}
\]
In this formula, $p_1,p_2$ are the enumeration parameters of this generating set. Note that $p_1$ is numerical (and ranges over the indices of the color classes), while $p_2$ is a domain variable.
The pairs of variables $b_sx_s$ and $b_tx_t$ are used to represent permutations in $\Sym(A\times A)$ as in \cref{def:ord_operator}.
For any $i,a\in \oA\times A$, $\mathfrak A\models \genG(i,a,\vec s,\vec t)$ if $s_1 = t_1$ and $t_2 = (\mapMa^i_{s_1})^{-1}\cdot\mapMa^i_a(s_2)$, i.e., if $\vec t = \iota(\phiRepr_i(\mapMa^i_a))(\vec s)$, which yields the desired result.
\end{proof}
\begin{theorem}
    \label{thm:minit_vinit}
    There is an $\FPC$-definable morphism $\minit : \mathcalG \to \gGa^{A\times A}$ and an $\FPC$-definable value $\vinit\in \gGa^{A\times A}$, such that $\phiRepr(\pi\Gamma) = \phiRepr(\pi)\psiRepr(\gGa) = \{\lambda\in\mathcalG, \minit(\lambda) = \vinit\}$.
\end{theorem}
That is, we prove that $\psiRepr(\gGa)$ is morphism-definable from $\mathcalG$ in $\FPC$ (by the morphism $\minit$), and provide a $\FPC$-definable value ($\vinit$) which represent its coset $\phiRepr(\pi\gGa)$ (w.r.t. the morphism $\minit$).
\begin{proof}[Sketch of proof] We define $\minit$ and $\vinit$ as follows:
\AP
\phantomintro\minit
\phantomintro\vinit
\begin{align*}
    \reintro *\minit(\lambda)_{a,b} &:= \begin{cases}
\lambda_a\lambda_b^{-1} &\text{if }\exists i, \{a,b\}\subseteq \ldomA_i\\
\Id &\text{otherwise}
\end{cases}\\
(\reintro *\vinit)_{a,b} &:= \begin{cases}
    (\mapMa^i_a)^{-1}\mapMa^i_{b} &\text{if }\exists i, \{a,b\}\subseteq \ldomA_i\\
\Id &\text{otherwise}
\end{cases}
\end{align*}
\Cref{lem:minit_morphism,lem:initMorph_correct,lem:initMorph_FP} in appendix show, respectively, that $\minit$ is indeed a morphism, that $\minit^{-1}(\vinit) = \phiRepr(\pi\gGa)$, and that $\minit$ and $\vinit$ are $\FPC$-definable, which altogether proves the theorem. Note that \cref{lem:minit_morphism} requires $\Gamma$ to be abelian.
\end{proof}
This concludes the morphism-definability of $\mathcal C_0$. It is only left to show that $\mathcal C^{i,j}_{a,b}$ is morphism-definable from $\mathcal G$, since intersections can be defined using \cref{corol:morph_def_grp_intersection}, and whether a coset represented by $(m,v)$ is empty can be defined using \cref{lem:fpc_def_image_morphism} (as this is equivalent to $v\not\in\im(m)$).
\begin{theorem}
    \label{thm:ext_semi_local_morphisms}
    For each $i <j\le m$, there is a $\FPC$-definable morphism
    \[\vartheta_{i,j} : \mathcalG\to \Sym(A\times A\times\slOm_{i,j})\]
    and a $\FPC$ definable function $\gslv_{i,j} : \ldomA_i\times \ldomA_j\to \Sym(A\times A\times \slOm_{i,j})$ such that, for any $(a,b)\in \ldomA_i\times \ldomA_j$ and $\sigma\in \pi\gGa$,
    \[ 
        \vartheta_{i,j}(\phiRepr(\sigma)) = \gslv_{i,j}(a,b)\iff \slE_{i,j}^{\sigma} =\slE_{i,j}^{\mapMa^i_a\oplus\mapMa^j_b}
    \]
\end{theorem}
The proof of this theorem relies on the existence of morphisms defining $\intro*\kDel_{i,j} := \lGa_i\lGa_j\cap\Aut(E_{i,j})$:
\phantomintro\slm
\begin{theorem}\label{thm:semi_local_morphisms}
    For each $i<j\le m$, there is a $\FPC$ definable ordered set $\intro*\slOm_{i,j}$, and a $\FPC$ definable morphism
    \[\reintro*\slm_{i,j} : \lGa_i\lGa_j \to \Sym(\slOm_{i,j})\]
    such that $\ker(\slm_{i,j}) = \kDel_{i,j}$.
\end{theorem}
\begin{proof}[Sketch of proof of \cref{thm:semi_local_morphisms}]
Because $\lGa_i$ and $\lGa_j$ are abelian, $\kDel_{i,j}\normle \lGa_i\lGa_j$. Consider the canonical morphism
\begin{align*}
    \theta_{i,j} : \lGa_i\lGa_j &\to (\lGa_i\lGa_j)/\kDel_{i,j}\\
    \LaTeXgamma &\mapsto \LaTeXgamma\kDel_{i,j}
\end{align*}
Because $\lGa_i$ and $\lGa_j$ are ordered by $\Phi$, we can represent each coset by its minimal representative. Let $\Omega_{i,j}$ be the set of those minimal representatives. Note that $|\Omega_{i,j}|$ is polynomially bounded by $|A|$ because $|\lGa_i\lGa_j|$ is. The action of $\Gamma_i\Gamma_j$ by left multiplication on $(\Gamma_i\Gamma_j)/\kDel_{i,j}$ defines an action of $\Gamma_i\Gamma_j$ on $\Omega_{i,j}$, and the corresponding morphism is definable within $\FP + \ord$. To summarize, for $\LaTeXgamma\in\lGa_i\lGa_j$ and $\omega\in\slOm_{i,j}$, $\slm_{i,j}(\LaTeXgamma)(\omega)$ is the unique $\omega'\in\slOm_{i,j}$ such that $\LaTeXgamma\omega\kDel_{i,j} = \omega'\kDel_{i,j}$. 
\end{proof}

\begin{proof}[Sketch of proof of \cref{thm:ext_semi_local_morphisms}]
    We provide the definition of $\vartheta_{i,j}$ and $\gslv_{i,j}$. For readability purposes, we present $\vartheta_{i,j}(\sigma)$ and $\gslv_{i,j}(a,b)$ as elements of $\Sym(\slOm_{i,j})^{A\times A}$ rather than $\Sym(A\times A\times\slOm_{i,j})$, relying on \cref{lem:enc_prod_groups} as in the previous subsections.
\phantomintro\vartheta
\begin{align*}
    \reintro*\vartheta_{i,j}: \mathcalG&\to \Sym(\slOm_{i,j})^{A\times A}\\
    \lambda&\mapsto \left((a,b)\mapsto
    \begin{cases}
        \slm_{i,j}(\lambda_{a}\lambda_{b})&\text{if }a\in \ldomA_i, b\in \ldomA_j\\
        \Id &\text{otherwise}
    \end{cases}\right)
\end{align*}
For $(a_i,a_j)\in \ldomA_i\times \ldomA_j$, let $\gslv_{i,j}(a_i,a_j)\in \Sym(\slOm_{i,j})^{A\times A}$ be defined, for any $(a,b)\in \ldomA_i\times \ldomA_j$, as
\[\intro*\gslv_{i,j}(a_i,a_j)(a,b) := 
    \slm_{i,j}((\mapMa^i_{a}\oplus\mapMa^j_{b})^{-1}\mapMa^i_{a_i}\oplus\mapMa^j_{a_j})\]
When $(a,b)\not\in\ldomA_i\times\ldomA_j$, we set $\gslv_{i,j}(a_i,a_j)(a,b)$ to $\Id$.
In \cref{lem:vartheta_morphism,lem:vartheta_correct,lem:vartheta_FP} in appendix, we show, respectively, that $\vartheta_{i,j}$ is a morphism, that it behaves as expected regarding to the encoding of $\slE_{i,j}$, and that both $\vartheta_{i,j}$ and $\gslv_{i,j}$ are $\FPC$-definable, which altogether yield the theorem.
\end{proof}
This concludes the definition of our representation of labeling cosets.
An important corollary of \cref{thm:ext_semi_local_morphisms} is that the value of $\gslv_{i,j}$ does not depend on the choice of pair $(a,b)$ within its equivalence class:
\begin{corollary}
    \label{corol:slMorph_value_invariant}
    For any two pairs $(a,b),(a',b')\in\ldomA_i\times\ldomA_j$, if $\slE_{i,j}^{\mapMa^i_a\oplus\mapMa^j_b} = \slE_{i,j}^{\mapMa^i_{a'}\oplus\mapMa^j_{b'}}$, then $\gslv_{i,j}(a,b) = \gslv_{i,j}(a',b')$.
\end{corollary}
This proves that our representation scheme is indeed isomorphism-invariant.
\Cref{thm:ord_abelian_colors} follows, with its corollary \cref{corol:rk_lt_ord}. 
We have left some technicalities out of this presentation. In particular, the precise way in which we define this procedure as a fix-point computation requires a few more steps: morphisms are defined by formulae with free second-order variables, not by relations. Indeed, as third order objects, they cannot be seen as variables within a fix-point computation. Instead, we define them syntactically, and only pass the value of $v$ in the representation $(m,v)$ of the current labeling coset. Another second-order variable is used to indicate which morphisms amongst $\minit\cup\{\vartheta_{i,j}\}$ are relevant at the current step of computation.

Another omission concerns the internal edges of a color-class: in the present proof, we have only treated the harder case of edges between two distinct color-classes. Indeed, the edges within $A_i$ can be treated in a simpler way: it happens that we can assume $\mathcal O(A_i)$ to already canonize those edges. For if it does not, we can refine the ordering on $\ldomA_i$, by setting $a\prec a'$ if $E_i^{\mapMa^i_a} <_{lex}E_i^{\mapMa^i_{a'}}$. $\PhiAB$ can be adapted to this finer pre-ordering, since the natural action of $\Gamma_i\cap \Aut(A_i)$ is transitive on each of the new color-classes induced by this refinement. We iterate this refinement process until a fix-point is reached.

\section{Conclusion}
We have introduced $\ord$, an operator enabling the definition of permutation group properties within fixed-point logics. As expected, this operator generalizes the rank operator $\rk$, as shown in \cref{thm:ord_ge_rk}.
Perhaps more surprising is the fact that $\ord$ is strictly more expressive than $\rk$ (\cref{sec:rk_lt_ord}). 
Indeed, this implies that the order of a group represented by a definable generating set is not definable with $\FP + \rk$.
To prove that $\FP + \ord$ is more expressive than $\FP + \rk$, we have shown that it canonizes structures with abelian colors, by defining the computation of the group-theoretic algorithm for the canonization of graphs defined in~\cite{babaiCanonical1983}.
While this algorithm was already defined in a logical context (precisely, in the context of Choiceless Polynomial Time\cite{pakusaLinear2015}), the use of a group-theoretic operator enabled a purely group-theoretic representation of labeling cosets.
This opens the door to novel techniques towards canonization in fixed-point logics: if our representation of labeling cosets relied heavily on the underlying groups being abelian, ordered, and transitive, a more complex representation scheme might allow the relaxation of some of those assumptions.


\bibliographystyle{IEEEtran}
\bibliography{Thesis.bib}
\newpage
\appendices
\section{Proof of \cref{thm:ord_ge_rk}}
\label{app:rank}
Consider the set $\mathcal E := \{E_{\vec a}^v \tq \vec a\in I, v\in\mathbb F_p\}$ of vectors of $\mathbb F_p^I$ defined by:
        \[{E_{\vec a}^v}(\vec b) := \begin{cases}
        0_{\mathbb F_p}   &\text{if }\vec a\ne \vec b\\
        v             &\text{otherwise}
        \end{cases}\]
        Obviously, any vector $\vec X \in \mathbb F_p^I$ is of the form $\sum_{\vec a\in I} E_{\vec a}^{X_{\vec a}}$.
        As such, this constitutes a generating set for $\mathbb F_p^I$.
        We now show how to represent $\mathbb F_p^I$ as a permutation group in $\FPC$.  First, we represent $v\in \mathbb F_p$ as the numerical permutation 
        \[f_v(w) : w\mapsto \begin{cases} w &\text{if }w\ge p\\
    w + v\mod p &\text{otherwise}\end{cases}\]
    Let $K := A^{\vec\pi}$.
    $v\mapsto f_v$ defines a morphism from $\mathbb F_p$ to $\Sym(K)$, and can obviously be defined in $\FPC$. In the same way, we define a morphism from $\mathbb F_p^I$ to $\Sym(I\times K)$:
  \[\iota_I(\vec X)(\vec a,w) := (\vec a,f_{X_{\vec a}}(w))\]
    We can define a similar representation $\iota_J$ of $\mathbb F^J_p$ in $\Sym(J\times \oA)$ (note that those representations of products are defined in the same way as in \cref{lem:enc_prod_groups}).
    We can now define an enumeration of $\iota_I(\mathcal E)$:
    \[\varphi_I(\vec p\vec\lambda, \vec x\vec\mu,\vec y\vec\nu)\! :=\! \begin{lorcases}
    \vec p = \vec x = \vec y\land \vec\mu<\vec\pi \land (\vec\nu\!=\! \vec\mu \!+\!\vec\lambda\!\!\!\!\mod\vec\pi)\\
    \vec p = \vec x = \vec y\land \vec\mu\ge\vec\pi\land \vec \nu = \vec \mu\\
    \vec p \ne \vec x = \vec y \land \vec \nu = \vec \mu
    \end{lorcases}\]
    Note that the tuple $\vec\pi$ is free in $\varphi_I$.
    For any $\vec a,v$, $\varphi_I(\mathfrak A,\vec a,\vec\lambda \gets v)$ defines the graph of $\iota_I(E^v_{\vec a})$.
    We can now define in the same way $\iota_I(f_{\mathcal M, p}(E^v_{\vec a}))$:
    \[\varphi_{\im_p(\mathcal M)}(\vec p\vec\lambda, \vec s\vec \mu, \vec t\vec \nu) :=
    \begin{lorcases}
        \vec\mu \ge \vec\pi \land \vec \mu = \vec\nu \land \vec s = \vec t\\
        \exists \vec m,
        \begin{landcases}
            \vec\mu < \vec\pi\land\vec\nu < \vec\pi\\
            \vec s = \vec t\\
            \varphi_M(\vec p,\vec s,\vec m)\\
            \vec\nu = (\vec\mu + \vec\lambda\cdot\vec m)\mod\vec\pi
            \end{landcases} 
    \end{lorcases}
    \]
    Let us walk through all variables appearing in this formula. $\vec p$ tracks the $I$-component $\vec a$ of the unit vector $E^v_{\vec a}$ under consideration. The value of $v$ is tracked by $\vec \lambda$. 
    $\vec s$ and $\vec t$ range over $J$, and represent the component of the vector $\mathcal M\cdot E^v_{\vec a}$ we are currently defining. Per our representation of group products, $\varphi_{\im_p(\mathcal M)}$ holds only if $\vec s = \vec t$. $\vec m$ tracks the value of the matrix $\mathcal M$ at coordinates $\vec a,\vec b$ where $\vec a$ is the current value of $\vec p$ and $\vec b$ the current value of both $\vec s$ and $\vec t$.

    Therefore, $\GspanFP{\varphi_{\im_p(\mathcal M)}}{\vec p\vec\lambda . \vec s\vec\mu.\vec t\vec\nu}{\mathfrak A} = \im_{\mathbb F_p}(\mathcal M)$, and our reasoning at the beginning of this proof yields the desired result.\qed

\section{Proofs omitted in \cref{sec:rk_lt_ord}}

\subsection{Proof of \cref{lem:local_labellings}}
\label{app:local_labellings}
Consider the formula $\Psi$ defined as follows:
\[ \Psi(\lambda,x,y,\mu) :=  \PhiAB(\lambda,\mu,x,y)\]
    $(\mathfrak A,i,a,b,k)\models \Psi$ iff $b$ is the $k$-th element in the ordering of $\ldomA_i$ defined by \cref{eqn:local_ordering_def}.
    However, this defines a bijection between $A_i$ and $[|A_i|]$.
    In order to build a bijection whose image is $\ldomAO_i$ instead of $\range{\ldomA_i}$, we add an adequate offset:
    \[\mapFP(\lambda,x,y,\mu) := \exists \nu, \left(\mu = \nu + \sum_{\lambda < \lambda} |A_\lambda|\right) \land \PhiAB(\lambda,\nu,x,y)\]
    The definability of those arithmetic operations in $\FPC$ is straight-forward. We have successfully built a formula $\mapFP$ such that, for all $i\le m$ and $a\in \ldomA_i$, $\mapFP(\mathfrak A,i,a)$ is the graph of a bijection between $\ldomA_i$ and $\ldomAO_i$, and thus "defines a partial ordering over $\ldomA_i$@definable partial ordering".
\subsection{Proof of \cref{thm:minit_vinit}}
\begin{lemma}
    \label{lem:minit_morphism}
    $\minit : \mathcalG \to \gGa^{A\times A}$ is a morphism
\end{lemma}
\begin{proof}
    Consider $\lambda,\lambda'\in \mathcalG$.
    For any $i\ne j$, $a\in \ldomA_i$ and $b\in \ldomA_j$, $\minit(\lambda\lambda')_{a,b} = \Id = (\minit(\lambda)_{a,b})(\minit(\lambda')_{a,b})$. 
    
    For $a,b\in \ldomA_i$,
    \begin{align*}
    \minit(\lambda\lambda')_{a,b} &= (\lambda\lambda')_a(\lambda\lambda')_b^{-1}\\
    &= \lambda_a\lambda'_a{\lambda'_b}^{-1}\lambda_b^{-1}\\
    \text{Since $\Gamma_i$ is abelian, }&= \lambda_a\lambda_b^{-1}\lambda'_a{\lambda'_b}^{-1}\\
    &= \minit(\lambda)_{a,b}\cdot \minit(\lambda')_{a,b}
    \end{align*}
\end{proof}
\begin{lemma}
    \label{lem:initMorph_correct}
    $\phiRepr(\pi\gGa) = \{ \lambda\in \mathcalG, \minit(\lambda) = \vinit\}$
\end{lemma}
\begin{proof}
For $\sigma\in\pi\gGa$ and $a,b\in \ldomA_i$, we have:
\begin{align*}
\minit(\phiRepr(\sigma))_{a,b} &= \phiRepr(\sigma)_a\phiRepr(\sigma)_b^{-1}\\
&=(\mapMa^i_a)^{-1}\sigma\sigma^{-1}\mapMa^i_b\\
&=(\vinit)_{a,b}
\end{align*}
We now show the other inclusion. Let $\lambda\in\mathcalG$ such that $\minit(\lambda) = \vinit$. Fix, for all $i$, some $a_i\in \ldomA_i$, and let $\sigma_i := \mapMa^i_{a_i}\lambda_{a_i}$. Then, for any $b\in \ldomA_i$,
\begin{align*}
    \mapMa^i_{b}\lambda_b &=\mapMa^i_b\lambda_b\lambda_{a_i}^{-1}\lambda_{a_i}\\
    &=\mapMa^i_b \minit(\lambda)_{b,a_i}\lambda_{a_i}\\
    &=\mapMa^i_b(\mapMa^i_b)^{-1}\mapMa^i_{a_i}\lambda_{a_i}\\
    &= \sigma_i
\end{align*}
Thus, for any $b$, $\lambda_b = (\mapMa^{i(b)}_b)^{-1}\sigma_{i(b)} = \phiRepr_{i(b)}(\sigma_{i(b)})_b$, which in turn implies that $\lambda = \phiRepr(\sigma_1\dots\sigma_m)$.
\end{proof}
In order to state \cref{lem:initMorph_FP}, we need a representation of $\Sym(A)^{A\times A}$ as a permutation group:
\begin{align*}
\iota_2 : \Sym(A)^{A\times A} &\to \Sym(A\times A\times A)\\
(\sigma_{(a,b)})_{(a,b)\in A\times A}&\mapsto \left((a,b,c)\mapsto (a,b,\sigma_{(a,b)}(c))\right)
\end{align*}
\begin{lemma}\AP
    \label{lem:initMorph_FP}
    \leavevmode
    \begin{itemize}
    \item There is a formula $\intro *\initMorph(R,a_s,b_s,x_s,a_t,b_t,x_t)$, such that, for any $\lambda\in\mathcal G$, $\initMorph(\mathfrak A,\graph(\iota(\lambda))) = \graph(\iota_2(\minit(\lambda)))$.
    \item There is a formula $\intro *\initValue(a_s,b_s,x_s,a_t,b_t,x_t)$ that defines in $\mathfrak A$ the graph of $\iota_2(\vinit)$.
    \end{itemize}
\end{lemma}
\begin{proof}
    \[
        \reintro *\initMorph(R,a_s,b_s,x_s,a_t,b_t,x_t):=
        \exists y, 
        \begin{landcases}
        a_s = a_t\\
        b_s = b_t\\
        R(b_s,x_t,y)\\ 
        R(a_s,x_s,y)
        \end{landcases}
    \]
    \[
        \reintro *\initValue(a_s,b_s,x_s,a_t,b_t,x_t) := \exists i,\mu, \begin{landcases}
        a_s = a_t\\
        b_s = b_t\\
        \mapFP(i,a_s,x_t,\mu)\\ 
        \mapFP(i,b_s,x_s,\mu)
        \end{landcases}\qedhere
    \]
\end{proof}

\subsection{Proof of \cref{thm:semi_local_morphisms}}
We introduce two intermediate lemmas.
First, we show that the membership to $\kDel_{i,j}$ is definable in $\FPC$:
\begin{lemma}
    There is a $\FPC$ formula $\intro*\aut(i,j,\mu,\nu)$ such that $\mathfrak A\models \aut(i,j,\alpha,\beta)$ iff $\gamma^i_\alpha\gamma^j_\beta\in \kDel_{i,j}$.
\end{lemma}
\begin{proof}
:
\[\reintro *\aut(i,j,\mu,\nu) := \forall a_i\in \ldomA_i, a_j\in \ldomA_j,\exists b_i,b_j, \begin{landcases}
    \PhiAB(i,\mu,a_i,b_i)\\
    \PhiAB(j,\nu,a_j,b_j)\\
    \Xi(a_i,a_j,b_i,b_j)
\end{landcases}
\]
where
\[\Xi(a_i,a_j,b_i,b_j) := E(a_i,a_j)\iff E(b_i,b_j)\qedhere\]
\end{proof}
This in turn enables us to check if two elements of $\lGa_i\lGa_j$ belong to the same coset of $\kDel_{i,j}$:
\begin{lemma}
    \label{lem:slmorph}
    There is a $\FPC$ formula $\intro* \coset(i,j,\mu,\nu,\mu',\nu')$ such that
    \[\mathfrak A\models \coset(i,j,\alpha,\beta,\alpha',\beta') \iff \gamma^i_{\alpha}\gamma^j_\beta\kDel_{i,j} = \gamma^i_{\alpha'}\gamma^j_{\beta'}\kDel_{i,j}\]
\end{lemma}
\begin{proof}
For $\alpha,\beta,\alpha',\beta'$,
$\gamma^i_\alpha\gamma^j_\beta\kDel_{i,j} = \gamma^i_{\alpha'}\gamma^j_{\beta'}\kDel_{i,j}$
iff $\chi := (\gamma^i_{\alpha'}\gamma^j_{\beta'})^{-1}\gamma^i_\alpha\gamma^j_\beta\in\kDel_{i,j}$. Because $\lGa_i$ and $\lGa_j$
act on disjoint sets, they commute, thus $\chi\in \lGa_i\lGa_j$
and there is a pair $(\alpha'',\beta'')$ such that $\chi = \gamma^i_{\alpha''}\gamma^j_{\beta''}$. Because $\chi\in\lGa_i\lGa_j$, $\chi\in\kDel_{i,j}$ is equivalent to $\chi\in\Aut(\slE_{i,j})$, which yields the following formula:
\[\coset(i,j,\mu,\nu,\mu',\nu'):= \exists\mu'',\nu'',\begin{landcases}
    \gamma^i_{\mu''} = (\gamma^i_{\mu'})^{-1}\gamma^i_\mu\\
    \gamma^j_{\nu''} = (\gamma^j_{\nu'})^{-1}\gamma^j_\nu\\
    \aut(i,j,\mu'',\nu'')
\end{landcases}
\]
To ease reading, we have included two clauses of the form $\sigma = \tau^{-1}\rho$ which are not $\FPC$ formulae per se. However, when the graphs of $\sigma$, $\tau$ and $\rho$ are defined by $R_\sigma,R_\tau$ and $R_\rho$, respectively, this equality can easily be defined in $\FPC$.
\end{proof}
We are now ready to prove \cref{thm:semi_local_morphisms}:
\begin{proof}[Proof of \cref{thm:semi_local_morphisms}]
We can  choose as a representative of the coset $\sigma\kDel_{i,j}$ the lexicographically smallest pair $(\alpha,\beta)$ such that $\gamma^i_{\alpha}\gamma^j_{\beta}\in\sigma\kDel_{i,j}$.
The following formula holds for $(i,j,\alpha,\beta)$ iff $(\alpha,\beta)$ is such a pair for some coset:
\[\mathsf{witness}(i,j,\mu,\nu) := \forall \mu',\nu',\begin{lorcases} 
    \mu\nu\le_{lex}\mu'\nu'\\
    \lnot\coset(i,j,\mu,\nu,\mu',\nu')
\end{lorcases}\]
This shows that $\reintro*\slOm_{i,j} := \mathsf{witness}(\mathfrak A,i,j)$ is definable in $\FPC$.

We aim to define $\slm_{i,j}(\LaTeXgamma)(\alpha,\beta)$ as the unique $(\alpha',\beta')\in\slOm_{i,j}$ such that $\LaTeXgamma\gamma^i_{\alpha}\gamma^j_{\beta}\in(\gamma^i_{\alpha'}\gamma^j_{\beta'})\kDel_{i,j}$, that is, such that $(\gamma^i_{\alpha'}\gamma^j_{\beta'})^{-1}\LaTeXgamma\gamma^i_{\alpha}\gamma^j_{\beta}\in \kDel_{i,j}$.
Note that, rather than keeping $\slm_{i,j}$ undefined on $(\oA)^2\setminus \slOm_{i,j}$, we simply set it to act trivially.
This indeed defines a morphism, as the action of $\slm_{i,j}(\LaTeXgamma)$ on $\slOm_{i,j}$ is, by construction, isomorphic to the action of $\LaTeXgamma$ on the set of cosets of $\kDel_{i,j}$ in $\lGa_i\lGa_j$ by left multiplication. For the same reason, $\ker({\slm_{i,j}}) = \kDel_{i,j}$.
\AP
\[\intro *\slMorph(i,j,R_\LaTeXgamma,\mu_s\nu_s,\mu_t\nu_t):= \begin{lorcases}
    \begin{landcases}\lnot\mathsf{witness}(i,j,\mu_s,\nu_s)\\ \mu_s = \mu_t \\ \nu_s = \nu_t\end{landcases}\\
    \begin{landcases}
        \mathsf{witness}(i,j,\mu_s,\nu_s)\\
        \mathsf{witness}(i,j,\mu_t,\nu_t)\\
        \Xi
    \end{landcases}
\end{lorcases}
\]
where
\[\Xi := \exists \mu',\nu',
\begin{landcases}
    \gamma^i_{\mu'}\gamma^j_{\nu'} = (\gamma^i_{\mu_t}\gamma^j_{\nu_t})^{-1}\LaTeXgamma\gamma^i_\mu\gamma^j_\nu \\
    \aut(i,j,\mu',\nu')
\end{landcases}
\]
This corresponds exactly to our definition of a definable morphism, in the sense that for any $\LaTeXgamma\in \lGa_i\lGa_j$,
\[\slMorph(\mathfrak A,i,j,\graph(\gamma)) = \graph(\slm_{i,j}(\gamma))\]
Note that, while we usually prefer to keep second order variables on the left-most side, we did not do so this time to underline the fact that $\slMorph$ actually defines a family of morphisms, one for each pair of colors.

Finally, let us justify our use of permutations equalities in the definition of $\Xi$. As the graph $R_\LaTeXgamma$ of $\LaTeXgamma$ is a parameter of $\slMorph$, and for all $i,\mu$, the graph of $\gamma^i_\mu$ is given by $\PhiAB(i,\mu)$, the subformula $\gamma^i_{\mu'}\gamma^j_{\nu'} = (\gamma^i_{\mu_t}\gamma^j_{\nu_t})^{-1}$ is in $\FPC$, following the same idea as in \cref{lem:perm_equalities}.
\end{proof}
\section{Proof of \cref{thm:ext_semi_local_morphisms}}
\begin{lemma}
    \label{lem:vartheta_morphism}
    For all $i< j\le m$, $\vartheta_{i,j}$ is a morphism.
\end{lemma}
\begin{proof}
Consider $\lambda,\lambda'\in \mathcalG$ and $(a,b)\in A\times A$. We aim to show that
\[(\vartheta_{i,j}(\lambda)\vartheta_{i,j}(\lambda'))_{(a,b)} =
\vartheta_{i,j}(\lambda\lambda')_{(a,b)}\]
If $(a,b)\not\in \ldomA_i\times \ldomA_j$, both sides of this equation evaluate to $\Id$. Otherwise,
\begin{align*}
    (\vartheta_{i,j}(\lambda)\vartheta_{i,j}(\lambda'))_{(a,b)} &=
    \slm_{i,j}(\lambda_a\lambda_b)\slm_{i,j}(\lambda'_a\lambda'_b)&\\
    &= \slm_{i,j}(\lambda_a\lambda_b\lambda'_a\lambda'_b)&\text{$m_{i,j}$ morph.}\\
    &= \slm_{i,j}(\lambda_a\lambda'_a\lambda_b\lambda'_b) &\text{ $\lGa_i\lGa_j$ ab.}\\
    &= \vartheta_{i,j}(\lambda\lambda')_{(a,b)}&\qedhere
\end{align*}
\end{proof}

\begin{lemma}
    \label{lem:vartheta_correct}
For any $i< j$, $(a_i,a_j)\in \ldomA_i\times \ldomA_j$, and $\sigma\in\pi\gGa$,
\[\vartheta_{i,j}(\phiRepr(\sigma)) = \gslv_{i,j}(a_i,a_j)\iff \slE_{i,j}^{\sigma} = \slE_{i,j}^{\mapMa^i_{a_i}\oplus\mapMa^j_{a_j}}\]
\end{lemma}
\begin{proof}
    First, notice that, for any $i<j$ and $a_i\in \ldomA_i, a_j\in \ldomA_j$, $\gslv_{i,j}(a_i,a_j) = \vartheta_{i,j}(\phiRepr(\sigma))$, where $\sigma$ is any element of $\pi\gGa$ such that $\sigma_{\restriction \ldomA_i\cup \ldomA_j} = \mapMa^i_{a_i}\oplus\mapMa^j_{a_j}$.

    It is thus only left to prove that, for two labellings $\sigma,\tau\in\pi\gGa$, $\vartheta_{i,j}(\phiRepr(\sigma)) = \vartheta_{i,j}(\phiRepr(\tau))$ iff $\sigma$ and $\tau$ yield the same encoding of $\slE_{i,j}$.
    \Cref{eqn:global_repr_coset} implies that
    \[\vartheta_{i,j}(\phiRepr(\sigma)) = \vartheta_{i,j}(\phiRepr(\tau)) \iff \vartheta_{i,j}(\psiRepr(\sigma^{-1}\tau)) = 1\]
    so that it is only left to show that, $\psiRepr(\LaTeXgamma) \in \ker(\vartheta_{i,j})\iff \LaTeXgamma\in\Aut(\slE_{i,j})$.
    And indeed:
    \begin{align*}
        \vartheta_{i,j}(\psiRepr(\LaTeXgamma)) \!=\! 1\!&\iff\! \forall a\in\ldomA_i,b\in\ldomA_j, \slm_{i,j}(\psiRepr(\LaTeXgamma)_a\psiRepr(\LaTeXgamma)_b)\!=\! 1\\
        \text{by def. of }\psiRepr,\!&\iff\! \forall a\in\ldomA_i,b\in\ldomA_j,\slm_{i,j}(\LaTeXgamma_{\restriction \ldomA_i}\LaTeXgamma_{\restriction \ldomA_j}) = 1\\
        &\iff\! \LaTeXgamma_{\restriction \ldomA_i\cup\ldomA_j}\in\ker(\slm_{i,j})\\
        \text{by def. of }\slm_{i,j},\!&\iff\! \LaTeXgamma_{\restriction \ldomA_i\cup\ldomA_j} \in\Aut(\slm_{i,j})
    \end{align*}
\end{proof}

To prove \cref{thm:ext_semi_local_morphisms}, it remains to show that $\vartheta_{i,j}$ and $\gslv_{i,j}$ are $\FPC$-definable. Once again, we use a representation of $\Sym(\slOm_{i,j})^{A\times A}$ as a permutation group. Recall that for each $i,j$, $\slOm_{i,j}\subseteq (\domAO)^2$.
\begin{align*}
    \iota_3 : \Sym((\domAO)^2)^{A\times A} &\to \Sym(A\times A\times ((\domAO)^2))\\
    (\sigma_{(a,b)})_{(a,b)\in A\times A}&\mapsto \left((a,b,(\mu,\nu))\mapsto (a,b,\sigma_{(a,b)}(\mu,\nu))\right)
    \end{align*}

\begin{lemma}\leavevmode
    \label{lem:vartheta_FP}
    There is a $\FPC$ formula $\varthetaFP(\mu,\nu,R,\vec s,\vec t)$ with $\type(\vec s) = \type(\vec t) = \mathrm{element}^2\mathrm{number}^2$ and $\type(R) = \mathrm{element}^2$ such that, for any $\sigma\in\phiRepr(\pi\gGa)$ and $i\ne j\le m$, 
    \[\varthetaFP(\mathfrak A,i,j,\graph(\iota(\sigma))) = \graph(\iota_3(\vartheta_{i,j}(\sigma))).\]
    There is a $\FPC$ formula $\vFP(\mu,\nu,x,y,\vec s,\vec t)$ with $\type(\vec s) = \type(\vec t) = \mathrm{element}^2\mathrm{number}^2$ such that, for any $i\ne j \le m$, 
    \[\vFP(\mathfrak A,i,j,a,b) = \graph(\iota_3(\gslv_{i,j}(a,b))).\]
\end{lemma}
\begin{proof}
Let us first define $\gslv_{i,j}$ as a formula $\vFP(\mu,\nu,x,y,\vec s,\vec t)$.
The variables $\mu,\nu$ track the pair of color-classes we are currently handling and $x,y$ track the component of $\gslv_{i,j}$ we are defining, i.e. we aim to obtain $\vFP(\mathfrak A,i,j,a_i,a_j) = \graph(\gslv_{i,j}(a_i,a_j))$.
Note that $\gslv_{i,j}(a_i,a_j)$ is represented as acting on $A\times A\times \slOm_{i,j}$.
Thus, the tuples of variables $\vec s$ and $\vec t$ are meant to represent individual elements of this set, and for readability purposes, we name those variables in accordance with our definition of $\gslv_{i,j}$ above: $\vec s = (a_s,b_s,\mu_s,\nu_s)$ and $\vec t = (a_t,b_t,\mu_t,\nu_t)$ ; with $(\mu_s,\nu_s)$ and $(\mu_t,\nu_t)$ representing elements of $\slOm_{i,j}$.
\[\intro *\vFP(\mu,\nu,x,y,\vec s, \vec t) :=
\begin{lorcases}
    (x\not \in A_\mu\lor y\not\in A_\nu)\land \vec s = \vec t\\
    (a_s\not\in A_\mu\lor b_s\not\in A_\nu)\land \vec s = \vec t\\
    \begin{landcases}x\in A_\mu\land y\in A_\nu\\
         a_s\in A_\mu\land b_s\in A_\nu \\
         a_s = a_t\land b_s = b_t\\
         \slMorph(\mu,\nu,R_g,\mu_s\nu_s,\mu_t\nu_t)\\
         \qquad\qquad\qquad[R_g(\alpha,\beta) / \xi(\alpha,\beta)]
    \end{landcases}
\end{lorcases}
\]
where
\[\xi(\alpha,\beta) := \begin{lorcases}
\begin{landcases}
    \alpha\in A_\mu\\
\beta\in A_\mu\\ \exists\lambda, \mapFP(\mu,x,\alpha,\lambda)\\ \mapFP(\mu,a_s,\beta,\lambda)
\end{landcases}\\
\begin{landcases}\alpha\in A_\nu\\
\beta\in A_\nu \\ \exists \lambda, \mapFP(\nu,y,\alpha,\lambda)
\\ (\mapFP(\nu,b_s,\beta,\lambda))
\end{landcases}
\end{lorcases}\]
Recall that $\mapFP$ is the formula defining the local-labellings $\mapMa^i_a$ as shown in \cref{lem:local_labellings}, that $\slMorph$ is the formula defining the morphisms $\slm_{i,j}$, as shown in \cref{lem:slmorph}, and that $F[G / H]$ denotes substitution, as explained in the proof of \cref{lem:fpc_def_image_morphism}.
Thus, for any suitable $\vec a = (i,j,a_i,a_j,a,b)$, an assignment of the free variables $(\mu,\nu,x,y,a_s,b_s)$ of $\xi$, $\xi(\mathfrak A,\vec a)$ defines the graph of $(\mapMa^i_a\mapMa^j_b)^{-1}\mapMa^i_{a_i}\mapMa^j_{a_j}\in \lGa_i\lGa_j$.
As such, the last clause in the definition of $\vFP$ correctly defines the graph (on the variables $\mu_s\nu_s,\mu_t\nu_t$) of $\slm_{i,j}((\mapMa^i_a\oplus\mapMa^j_b)^{-1}(\mapMa^i_{a_i}\oplus\mapMa^j_{a_j}))$.

We now turn to the definition of the morphisms $\vartheta_{i,j}$. We provide a formula $\varthetaFP(\mu,\nu,R,\vec s, \vec t)$, where once again, $(\mu,\nu)$ tracks the pair of color-classes at hand, and now, $R$ should be the graph of an element $g\in \iota(\mathcalG)$.

As in the definition of $\vFP$, $\vec s$ and $\vec t$ represent the domain $A\times A\times \slOm_{i,j}$ of the permutational image of the morphism $\vartheta_{i,j}$, and we name each individual variable $a_s,b_s,\mu_s,\nu_s$ (resp. $a_t,b_t,\mu_t,\nu_t$) to improve readability.
\[\intro*\varthetaFP(\mu,\nu,R,\vec s,\vec t) := \begin{landcases}
    a_s = a_t\\
    b_s = b_t \\
    \slMorph(\mu,\nu,R_g,\mu_s\nu_s,\mu_t\nu_t)\\
    \qquad\qquad\qquad[R_g(x,y) / \theta(x,y)]
\end{landcases}
\]
where 
\[\theta(x,y) := \exists w, R(a_s,x,a_s,w)\land R(b_s,w,b_s,y)\]
One should be careful not to confuse $R_g$, which is the second-order variable in $\slMorph$ that should take as value the graph of a permutation in $\lGa_i\lGa_j$, and $R$, which is used in $\vartheta$ as a place-holder for the graph of a permutation in $\mathcalG$, so that $\vartheta(\mu,\nu,R,\vec s,\vec t)$ defines a morphism as explained in \cref{def:definable_morphism}.
\end{proof}
\end{document}